\documentclass[a4paper,UKenglish,cleveref, autoref, thm-restate]{lipics-v2021}

\usepackage[linesnumbered,ruled,vlined]{algorithm2e}
\usepackage{tikz}
\usetikzlibrary[shapes, arrows, positioning]

\newcommand{\trp}{\textsc{trp}}
\newcommand{\texp}{\textsc{texp}}

\newcommand{\autg}{\ensuremath{\text{Aut}(G)}}
\newcommand{\auttg}{\ensuremath{\text{Aut}(\mathcal{G})}}

\newcommand{\torb}{\ensuremath{\mathcal{G}/\auttg}}

\long\def\full#1{
{#1}%
}
\long\def\conf#1{%
{}%
}
\long\def\app#1{%
{}%
}

\hideLIPIcs  

\title{Exploiting Automorphisms of Temporal Graphs for Fast Exploration and Rendezvous}

\author{Konstantinos Dogeas}{Department of Computer Science, Durham University, Durham, England}{konstantinos.dogeas@durham.ac.uk}{https://orcid.org/????-????-????-????}{}

\author{Thomas Erlebach}{Department of Computer Science, Durham University, Durham, England}{thomas.erlebach@durham.ac.uk}{https://orcid.org/0000-0002-4470-5868}{}

\author{Frank Kammer}{THM, University of Applied Sciences Mittelhessen,
Gie\ss en, 
Germany}{frank.kammer@mni.thm.de}{https://orcid.org/0000-0002-2662-3471}{}

\author{Johannes Meintrup}{THM, University of Applied Sciences Mittelhessen, 
Gie\ss en, 
        Germany}{johannes.meintrup@mni.thm.de}{https://orcid.org/0000-0003-4001-1153}{Funded by the Deutsche Forschungsgemeinschaft 
        (DFG, German Research Foundation) -- 379157101.}

\author{William K. Moses Jr.}{Department of Computer Science, Durham University, Durham, England}{william.k.moses-jr@durham.ac.uk}{https://orcid.org/0000-0002-4533-7593}{}

\authorrunning{K. Dogeas, T. Erlebach, F. Kammer, J. Meintrup, and W.K. Moses Jr.}

\authorrunning{\ } 

\Copyright{\ } 



\category{} 

\relatedversion{} 

\nolinenumbers 

\begin{document}
\fontsize{11}{13}\selectfont

\maketitle

\begin{abstract} 
Temporal graphs are dynamic graphs where the edge set can change in each time
step, while the vertex set stays the same. Exploration of temporal graphs whose
snapshot in each time step is a connected graph, called connected temporal
graphs, has been widely studied. In this paper, we extend the concept of graph
automorphisms from static graphs to temporal graphs and show for the first
time that symmetries enable faster exploration: We prove that a connected temporal
graph with $n$ vertices and orbit number $r$ (i.e., $r$~is the number of
automorphism orbits) can be explored in $O(r n^{1+\epsilon})$ time steps, for any
fixed $\epsilon>0$. For $r=O(n^c)$ for constant $c<1$, this is a significant
improvement over the known tight worst-case bound of $\Theta(n^2)$ time steps
for arbitrary connected temporal graphs. We also give two lower bounds for
temporal exploration, showing that $\Omega(n \log n)$ time steps are required
for some inputs with $r=O(1)$ and that $\Omega(rn)$ time steps are required for
some inputs for any $r$ with $1\le r\le n$.

Moreover, we show that the techniques we develop for fast exploration can be
used to derive the following result for rendezvous: Two agents with different
programs and without communication ability are placed by an adversary at
arbitrary vertices and given full information about the connected temporal
graph, except that they do not have consistent vertex labels. Then the two
agents can meet at a common vertex after $O(n^{1+\epsilon})$ time steps, for any
constant $\epsilon>0$. For some connected temporal graphs with the orbit number
being a constant, we also present a complementary lower bound of $\Omega(n\log
n)$ time steps. 
\full{{Finally, we give a randomized algorithm to construct a temporal
walk $W$ that visits all vertices of a given orbit with probability at least $1-\epsilon$
for any $0<\epsilon<1$
such that $W$ spans $O((n^{5/3}+rn)\log n)$ time steps. The runtime of the
algorithm consists of $O(n^{1/3} \log (n/\epsilon))$ linear-time
{considerations} of the
snapshots that exist in the aforementioned time span.
We also give an alternative algorithm that has an expected runtime
consisting of $O(n^{1/3})$ linear-time {considerations}, but provides the specified
temporal walk with probability $1$.
}}
\end{abstract}

\section{Introduction}
\label{sec:int}

For many decades, graph theory has been a tool used to model and study many real
world problems and phenomena~\cite{BondyM76}. A usual assumption for many of
these problems is that the graphs have a fixed structure. However, there are
quite a number of cases where the structure of a system changes over time. For
example, consider the problem of routing in transportation networks (roads,
rails) where specific connections can become unavailable (due to a disaster) or
they are active during specific times (due to safety). Such scenarios can be
modeled using temporal graphs, a sequence of graphs over the same vertex set
where the edges possibly change in each time step. The temporal graph setting
has received significant interest from the research community in the recent past
as seen in recent surveys~\cite{CasteigtsFQS12,Michail16}. 

In this paper, we study two problems on temporal graphs. The first is the
\textit{temporal exploration problem} ($\texp$), which has been studied, e.g., by Michail and
Spirakis~\cite{MichailS14,MichailS16} and by Ilcinkas et
al.~\cite{IlcinkasKW14}. It requires an agent to explore all vertices of the
temporal graph as quickly as possible. The second is the \textit{temporal
rendezvous problem} ($\trp$), which we formulate for the first time in this
paper. It requires two heterogeneous agents (in terms of the programs they run)
to rendezvous on a temporal graph when they cannot communicate with one another.
For both problems, we assume that each  agent has complete knowledge of the
temporal graph in advance (a common
assumption~\cite{AdamsonGMZ22,BodlaenderV19,ErlebachHK21,ErlebachKLSS19,ErlebachS22a,ErlebachS22b,IlcinkasKW14,MichailS14,MichailS16,Taghian20}).
However, in the case of $\trp$ the agents may have different names for the same
vertices, i.e., the local labels of the vertices may be different. 

The problem of exploration has been well studied in the static setting since it was introduced in 1951 by Shannon~\cite{Shannon51}.  
It has also been intensively studied in the temporal graph setting since
2014 (see all references from the previous paragraph).
On an application-oriented note, the problem captures the setting where a person
is trying to visit various parts of a city using 
public
transportation. For example, train schedules involve multiple train
stations (vertices) with trains running between them at different times (i.e.,
repeatedly
changing edges). Thus, planning a visit to multiple destinations
over a given day using railways is an example of solving $\texp$.

The rendezvous problem can be broadly categorized into two types: symmetric and
asymmetric rendezvous. The version where agents have the same strategy
(symmetric rendezvous) was introduced by Alpern~\cite{Alpern76}. The version
where agents can have distinct strategies (asymmetric rendezvous) was introduced
by Alpern~\cite{Alpern95} and is the focus of research in this paper. 
$\trp$ is a natural
extension of the asymmetric rendezvous problem to the dynamic setting. As a real
world example, consider a pair of tourists who want to explore a city together
and have to agree on a strategy to meet up in case they are separated and their
cell phones die. In this scenario, they may use public transportation
(dynamically changing network) to meet and agree in advance to use different
strategies that guarantee that they meet quickly. 

In this paper, we present results that extend the literature in two ways.
Firstly, we formalize the \trp{} problem in the setting where agents have
complete knowledge of the temporal graph a priori and develop good upper and
lower bounds for it. 
Secondly, we
utilize interesting structural properties (namely automorphisms of a graph and
associated orbits of  vertices) in order to analyze 
bounds on the number of time steps required by temporal walks that solve certain problems.
In particular, we develop upper and lower bounds for both \texp{} and \trp{} that leverage the 
aforementioned graph properties. To the best of our knowledge, this is the first work that takes 
advantage of these graph properties to study problems in temporal graphs. 

\subsection{Our Contributions}
\label{subsec:our-contribs}
We present results for two problems: first, the \textsc{temporal exploration
problem} ($\texp$) which, for a given temporal graph $\mathcal{G}$, asks for a
temporal walk that visits all vertices of $\mathcal{G}$, and secondly, the
\textsc{temporal rendezvous problem} (\trp), which considers two agents that try
to meet in the given temporal graph, meaning they must be stationed at the same
vertex in the same time step.

One of our primary contributions is formalizing the problem of $\trp$ in
Section~\ref{sec:pre}, and in doing so extending the problem of asymmetric
rendezvous to the temporal graph setting where agents have complete knowledge of
the temporal graph in advance. {Another significant contribution is that we
show} how to leverage the use of a structural graph property, namely the
automorphism group of a temporal
graph and the associated notion of orbits, to bound the number of time steps
required by algorithms we devise for these problems. To the best of our
knowledge, this work is the first instance of leveraging such properties to
study any problem on temporal graphs apart from a practical
implementation of a generator for listing all non-isosomorphic simple temporal
graphs~\cite{Casteigts20}. Intuitively, an automorphism of a graph is a mapping
from the set of vertices of the graph to the same set of vertices that preserves
the
neighborhood relation between the vertices. The set of automorphisms of a temporal graph
consists of the \textit{intersection} of the sets of automorphisms of the
graph at each time step. The set of automorphisms of the temporal graph,
along with function composition as group operation, forms the automorphism group of the
temporal graph. An orbit of the automorphism group of a temporal graph is a maximal set
of vertices such that each vertex can be mapped to any other vertex in the set via
an automorphism of the group.
Intuitively, the vertices in an orbit look indistinguishable
to an agent that has full information about the temporal graph but
without meaningful vertex labels. As the agents are able to compute a
consistent numbering of the orbits (Lemma~\ref{lem:meeting}), they can agree to meet
in some specific orbit, but not at a specific vertex.
{Therefore, temporal graphs where all orbits are large
(or that even have a single orbit containing all vertices)
appear to be the most challenging graphs for solving
$\trp$. Our result providing fast exploration schedules
in temporal graphs with few orbits is therefore a
crucial ingredient for enabling our solution to $\trp$ to handle
all possible temporal graphs (including those with a single
orbit).}

%
%
%

We give precise definitions and present further preliminaries in
Section~\ref{sec:pre}. Then,
we introduce some useful
utilities related to automorphisms in Section~\ref{sec:aut-util}. These
will be used in later sections, where we present the following
results. 

\begin{sloppypar}
\textbf{Upper bounds.} 
In Section~\ref{sec:texp}, we develop a deterministic algorithm to solve $\texp$
in $O(r n^{1+\epsilon})$ time steps for any fixed $\epsilon >0$ (see
Corollary~\ref{cor:epsilonallorbits}), where $n$ is the number of vertices in
the temporal graph and $r$ is the number of orbits of the automorphism group of
the temporal graph. Note that $r$ can range in value from $1$ to $n$. Thus, for
$r=O(n^c)$ for constant $c<1$, this is a significant improvement over the known
tight worst-case bound of $\Theta(n^2)$ time steps for
arbitrary connected temporal graphs~\cite{ErlebachHK21}. 
In Section~\ref{sec:rendezvous}, {we leverage this algorithm for $\texp$ to
develop a deterministic solution for $\trp$ using} $O(n^{1+\epsilon})$ time
steps for any fixed $\epsilon > 0$ (see Theorem~\ref{thm:rendezvous-time}). 
Our focus is on bounding the time steps of the temporal walks required to
solve $\texp$ and $\trp$, and not on \mbox{optimizing the running time of the
respective algorithms to compute such walks.}
\end{sloppypar}

\textbf{Lower bounds.} 
We complement our algorithms with lower bounds for both $\texp$ and $\trp$ in
Section~\ref{sec:lower-bounds}. In particular, we design an instance of $\trp$
such that any solution for it requires $\Omega(n \log n)$ time steps (see
Theorem~\ref{thm:rendezvouslb}). We then show how this translates to a lower
bound of $\Omega(n \log n)$ time steps for {some instances of $\texp$ where
the temporal graph at hand has an orbit number $r=O(1)$ (see
Corollary~\ref{cor:explorationlb})}. By revisiting the lower bound of
\cite{ErlebachHK21} for $\texp$, which focused on arbitrary temporal graphs, and
studying it through the lens of automorphisms and orbits, we can obtain a more
fine-grained lower bound of $\Omega(rn)$ time steps (see
Lemma~\ref{lem:rn-lower-bound}) for {some temporal graphs with orbit number
$1\le r \le n$.} Notice that the {multiplicative} gap between our upper and lower bounds for both
problems is only {a factor of} $O(n^\epsilon)$.

{ \textbf{Relevance of the orbit number.} {It is well
known in graph theory}
    that almost all (static) graphs are \textit{rigid}~\cite{BallGS18},
    meaning that the automorphism group contains only the trivial
    identity function. This is in contrast to observations in practice.
    Many real-world graphs have non-trivial automorphisms. A recent
    analysis of real-world graphs in the popular database
    \texttt{networkrepository.com} showed that over $70\%$ of the analyzed
    graphs had non-trivial automorphisms~\cite{BallGS18}. One may
    reasonably expect that real-world temporal graphs have similar
    properties. Symmetries are also abundant in graphs arising in
    chemistry~\cite{Balasubramanian82}, which have been studied in temporal
    settings as well~\cite{PlamperLH23}.
	{We believe that these observations provide a strong
	motivation for studying temporal graph problems for temporal
	graphs with fewer than $n$ orbits,{\ e.g., with
    algorithms parameterized via the number of orbits.}}
	{Furthermore, we emphasize that
    our results for $\trp$ {hold for all temporal graphs,}
	independent of the orbit number parameter,
even {though} they {utilize} techniques we develop for $\texp$
parameterized by the orbit number.} {Roughly speaking,
for orbit number~$r$, the penalty factor $r$ in the number of time steps 
to solve $\texp$ is saved when solving $\trp$ by focusing only on a smallest orbit of size at most $n/r$.
}
}

\subsection{Technical Overview and Challenges}
\label{subsec:tech-overview-challenges}

Firstly, we give an intuitive overview of our upper bound results. The concepts
used in this section are more precisely defined in the next sections.

For $\texp$, we first consider the problem of visiting all the vertices of one
orbit~$S$. One key insight is that, if we have a temporal walk $W$ that visits
$k$ vertices of~$S$, we can use the automorphisms of the temporal graph to
transform $W$ into other walks that visit different sets of $k$ vertices of~$S$
(Lemma~\ref{lem:transform1}).
Therefore, even if the $k$ vertices visited by $W$ have already been explored
earlier, we can transform $W$ into a temporal walk $W'$ that visits a ``good''
number of previously unexplored vertices of~$S$. The number of previously
unexplored vertices of $S$ that $W'$ is guaranteed to visit increases with the
number of possible start vertices in $S$ that we allow for $W'$, but the larger
that set $X$ of possible start vertices is, the longer it may take the agent to
reach the best start vertex in that set. A challenge is to analyze this
tradeoff. By carefully relating the size of $X$ to the guaranteed number of
unexplored vertices that a walk $W'$ starting at a vertex in~$X$ can visit
(Corollary~\ref{cor:transform2}), we
manage to balance the number of time steps needed to move to the start vertex of
$W'$ and the number of previously unexplored vertices that $W'$ visits.
To show that vertices of $X$ can be reached quickly, we study the structure
of edges that connect vertices in different orbits (Lemma~\ref{lem:orbitedges})
and use it to analyze
reachability between orbits (Lemmas~\ref{lem:laneorbit} and~\ref{lem:nextreachable}).

We then employ a recursive construction: We recursively construct a temporal
walk {$W_{1}$ that} visits $k$ vertices of~$S$, concatenate
$W_{1}$ with a temporal walk that moves quickly from the endpoint of
$W_{1}$ to a good start vertex in $S$ for what follows, and then use a second recursively constructed walk
(transformed via an automorphism to a ``good'' temporal walk
{$W_{2}$} starting in a vertex of $S$) to visit ``nearly'' $k$
further vertices of~$S$. {A careful} analysis then shows that in this way
we can visit a constant fraction of the vertices of $S$ in
$O(r|S|^{1+\epsilon})$ steps (Lemma~\ref{lem:vtexplorationub} and Corollary~\ref{cor:tinyfraction}).
We then show that the concatenation of $O(\log
|S|)$ such walks (each one again transformed via an automorphism to maximize the
number of newly explored vertices) suffices to visit all vertices of $S$ in
$O(r|S|^{1+\epsilon}+n\log |S|)$ time steps (Theorem~\ref{thm:epsilonexploration}),
which can be bounded by $O(r
n^{1+\epsilon'})$ time steps). By visiting the $r$ orbits one after another, we
can finally show that the whole temporal graph can be explored in
$O(r n^{1+\epsilon'})$ time steps (Corollary~\ref{cor:epsilonallorbits}).

For $\trp$, a simple and fast solution one first thinks of is to have the agents
simply meet at a vertex with a specific label. {However}, this is not
feasible in the model we consider. In particular, we assume that the agents
cannot communicate 
and, while both agents have complete information about the
temporal graph, they do not have access to consistent vertex labels. As such,
the agents are unable to agree upon the same vertex based on the vertex label.
We rely, instead, on structural graph properties and let the agents meet at a
vertex in a smallest orbit: One agent moves to any vertex in that orbit, and the
other searches all vertices in that orbit. The first challenge is for the agents
to independently identify the same orbit in which to meet. We show that this can
be done by letting each agent enumerate all temporal graphs, together with
colorings of their orbits, until it encounters for the first time a temporal
graph that is isomorphic to the input graph in which $\trp$ is to be solved. In
this way both agents obtain the same colored temporal graph and can
independently select, among all orbits of smallest size, the one with smallest
color (Lemma~\ref{lem:meeting}).
Then one agent moves to a vertex in that orbit~$S$, while the other searches all
vertices of~$S$. The challenge here is to deal with orbits that are large,
because previously known techniques would require $\Theta(n^2)$ time steps to
explore an orbit of size $\Omega(n)$. {Fortunately}, this is where our
above-mentioned results for exploration of (orbits of) temporal graphs come to
the rescue: As $|S|\le n/r$, we have $O(r|S|^{1+\epsilon}+n\log
|S|)=O(n^{1+\epsilon'})$, and hence $S$ can be explored (and $\trp$ solved) in
$O(n^{1+\epsilon'})$ time steps (Theorem~\ref{thm:rendezvous-time}).

Now we turn to our lower bounds.
For $\texp$, we first observe that the
existing lower bound construction that shows that $\Omega(n^2)$ time steps are
necessary for exploration on some temporal graphs uses temporal graphs with
$\frac{n}{2}+1$ orbits. By slightly varying the construction, we can show
for any $r$ in the range from $1$ to $n$
that there are temporal graphs with $r$ orbits that require
$\Omega(rn)$ time steps for exploration (Lemma~\ref{lem:rn-lower-bound}).
Our main lower bound result is the lower bound of
$\Omega(n\log n)$ for $\trp$ in temporal graphs with a single orbit
(Theorem~\ref{thm:rendezvouslb}). The
temporal graph is a cycle {on $n=2^m-1$ vertices} in every step, and the
edges change every $\lfloor n/16\rfloor$ time steps. Each period of $\lfloor
n/16\rfloor$ time steps in which the edges do not change is called a
\emph{phase}. We number the vertices of the cycle from $0$ to $n-1$ and consider
the binary representation of the vertex labels. {In phase~$1$, each vertex
$j$ is adjacent to $j+1$ and $j-1$, while in any phase $i>1$, it is adjacent to
$j+2^{2i}$ and $j-2^{2i}$ (with all computations done modulo~$n$).} As all
vertices are indistinguishable to the agents, we can force each agent to make
the same movements no matter where we place it initially. By fixing the five
highest-order bits of the agents' start positions in a certain way, we can
ensure that the agents do not meet in the first phase, no matter how the
remaining bits of their start positions are chosen. Depending on the positions
where the agents end up at the end of each phase (relative to their start
position in that phase), {we fix a certain number of bits of the binary
representations of the start positions of the agents: After the first phase, we
fix the lowest four bits, and after each further phase, we fix the next-higher
two bits. We can show that this is sufficient to ensure} that the agents start
the next phase at vertices that are sufficiently far from each other in the
cycle of that phase. This process can be repeated $\Omega(\log n)$ times,
showing that it takes $\Omega(n \log n)$ {time steps} for the agents to meet. This lower
bound for $\trp$ implies also that exploration of the constructed temporal graph
requires $\Omega(n \log n)$ steps (Corollary~\ref{cor:explorationlb}).

\full{In Section~\ref{sec:rand} we present a randomized algorithm for
    computing a temporal walk $W$ that visits all vertices of a given
    orbit with probability at least $1-\epsilon$ for any $0<\epsilon<1$.
    The computed temporal walk $W$ spans $O((n^{5/3}+rn)\log n)$ time
    steps, with $r$ the number of orbits of the automorphism group. The
    walk can be computed by $O(n^{1/3} \log (n/\epsilon))$ linear-time
    scans of the snapshots that exist in the aforementioned time span. We
    also give an alternative algorithm that has an expected runtime that
    consists of $O(n^{1/3})$ linear-time scans, but provides the specified
    temporal walk with probability $1$. In contrast to the upper bounds
    presented in the preceding sections, the algorithm has a runtime that
is independent of the number of automorphisms of the temporal graph at
hand.}

\subsection{Related Work}
\label{subsec:related-work}

There is a divide in the research on temporal graphs with respect to the
amount of knowledge known in advance by the agents on the graph.  When
complete knowledge of the temporal graph is known in advance to the agents,
as is assumed in this paper, the setting is
sometimes referred to as the \emph{post-mortem}
setting (see, e.g., Santoro~\cite{Santoro15}). Since we consider
scenarios where agents plan temporal walks using advance knowledge
of future time steps, we refer to the post-mortem
setting as the {\em clairvoyant} setting in the remainder of this section.

The clairvoyant setting is in contrast to the \emph{live} setting where agents only have partial
knowledge of the temporal graph, and the solutions to problems in these
different settings are of a different nature.  The reader can refer to the
survey by Di Luna~\cite{DiLuna19} for more information on problems and
solutions in the live setting.  We restrict the rest of this related work
section to work in the clairvoyant setting.

~\\ \noindent \textbf{Exploration.} 
%
%
We trace the progress of solving instances of $\texp$, first by identifying
which assumptions are required for the problem to even be solvable, and then by
identifying properties that were leveraged to give faster and faster solutions. 

Michail and Spirakis~\cite{MichailS14,MichailS16} showed that it is
\textbf{NP}-complete to decide whether a given temporal graph with a given
\textit{lifetime}, i.e., the number of time steps it exists, is explorable when
no assumptions are made on the input graph. This holds even if the graph is connected in
every time step, termed \textit{connected} (and sometimes called \emph{always-connected}
in the literature).
Even when a restriction is placed on the \textit{underlying graph}, i.e., the
union of the graphs at each time step, such that the underlying graph has
pathwidth at most~$2$, Bodlaender and van der Zanden~\cite{BodlaenderV19} showed
that $\texp$ is \textbf{NP}-complete. 

However, the exploration is always possible when the lifetime of the graph is
sufficiently large. In particular, Michail and Spirakis~\cite{MichailS16} showed
that a connected graph with $n$ vertices may be explored in $O(n^2)$ time steps. For
the rest of the related work on $\texp$, we focus on the case of connected
temporal graphs with a sufficiently large lifetime. Erlebach et
al.~\cite{ErlebachHK21} showed that exploration on arbitrary temporal graphs
takes $\Omega(n^2)$ time steps. By restricting their study of temporal
graphs to those where the underlying graph belongs to a special graph class,
however, they showed that $\texp$ can be solved in $o(n^2)$ time steps in
several such cases. In particular,
when the underlying graph is planar, has bounded treewidth~$k$, or is a $2
\times n$ grid, they showed that $\texp$ can be solved in $O(n^{1.8} \log n)$
time steps, $O(n^{1.5} k^{1.5} \log n)$ time steps, and $O(n \log^3 n)$ time
steps, respectively. They also showed a lower bound of $\Omega(n \log n)$ when
the underlying graph is a planar graph of degree at most $4$. 

The study of $\texp$ when the temporal graph is restricted continued in
several papers. Taghian Alamouti~\cite{Taghian20} showed that $\texp$ can be
solved in $O(k^2 (k!) (2e)^k n)$ time steps when the underlying graph is a cycle
with $k$ chords. Adamson et al.~\cite{AdamsonGMZ22} improved this to $O(kn)$
time steps. They also improved the upper bounds on $\texp$ for underlying graphs
that have bounded treewidth $k$ or are planar to $O(k n^{1.5} \log n)$ and
$O(n^{1.75} \log n)$, respectively. In addition, they strengthened the lower bound for
underlying planar graphs by showing that even if the degree is at most~$3$, the
lower bound is $\Omega(n \log n)$. Erlebach et al.~\cite{ErlebachKLSS19} further
improved work on bounded degree underlying graphs by showing that  $\texp$ can
be solved in $O(n^{1.75})$ time steps for such temporal graphs.
Ilcinkas et al.~\cite{IlcinkasKW14}
showed that when the underlying graph is a cactus, the exploration time is
$2^{\Theta(\sqrt{\log n})} n$ time steps. 

Other variants of $\texp$ have been studied where the problem is slightly
different, or the edges of the temporal graph vary in some particular way (e.g.,
periodically, $T$-interval connected, $k$-edge deficient), or multiple agents
explore the
graph~\cite{AaronKM14a,AaronKM14b,AkridaMSR21,BumpusM23,ErlebachS22a,ErlebachS22b,IlcinkasW13}.

~\\ \noindent \textbf{Rendezvous.} 
%
Symmetric rendezvous~\cite{Alpern76} and asymmetric rendezvous ~\cite{Alpern95}
have received much interest over the years, resulting in numerous
surveys~\cite{Alpern02,AlpernG06,Flocchini19,KranakisKM22,KranakisKR06,Pelc12,Pelc19}
being written on the problems covering different settings. 
In this work, we are
the first to extend asymmetric rendezvous to temporal graphs in the clairvoyant
setting. Note that there has been previous
work~\cite{BournatDP18,DasDPP19,DiLunaFPPSV20,MichailST21,OoshitaD18,ShibataKESNK23,ShibataSNK21}
that has studied multi-agent rendezvous (also called gathering) in a dynamic
graph setting, but that was in the live setting. 
For an overview of rendezvous in static graphs we refer the reader to~\cite{AlpernG06,
KranakisKR06}. Rendezvous has been studied in deterministic
settings~\cite{BhagatP22,Pelc23} and asynchronous
settings~\cite{dieudonne2013meet,MARCO2006315}, and has been shown to be
solvable in logarithmic space~\cite{CzyzowiczJKAP10}.

~\\ \noindent \textbf{Other Related Work.}  
Other problems have also been studied in the temporal graph setting, e.g.,
matchings~\cite{MertziosMNZZ20},
separators~\cite{FLUSCHNIK2020197,ZSCHOCHE202072}, vertex
covers~\cite{AKRIDA2020108,HammKMS22}, containments of
epidemics~\cite{ENRIGHT202160}, Eulerian tours~\cite{BumpusM23,MarinoS21}, graph
coloring~\cite{MarinoS22,MertziosMZ21}, network flow~\cite{AKRIDA201946},
treewidth~\cite{Fluschnik2020} and cops and robbers~\cite{Erlebach-Spooner/20,MorawietzRW20}. See
the survey by Michail~\cite{Michail16} for more information.

\section{Preliminaries}
\label{sec:pre}

\textbf{Basic terminology.} We use standard graph terminology. All static graphs
are assumed to be simple and undirected. We write $[x]=\{1, \ldots, x\}$ for any
integer $x$.

A \textit{temporal graph} $\mathcal{G}=(G_1, \ldots, G_{\ell})$ is a sequence of
static graphs, all with the same vertex set~$V$. For a temporal graph
$\mathcal{G}$, the graph $G=(V, \bigcup_{t \in [\ell]} E(G_t))$ is the
\textit{underlying graph of $\mathcal{G}$}. We call $\mathcal{G}$ connected if
each $G_t$ for $t\in [\ell]$ is connected. We use $n$ to refer to the number of
vertices of the (temporal) graph under consideration. A \textit{temporal walk}
$W$ is a walk in a temporal graph $\mathcal{G}$ that is time respecting, i.e.,
traverses edges in strictly increasing time steps. For a temporal walk $W$
starting at time step $t$ we write $W=(u_1, u_2, u_3, \ldots)$ to mean the
temporal walk starts at vertex $u_1$ in time step $t$, {is at vertex $u_2$
at the beginning of time step $t+1$,} and so forth. Note that subsequent
vertices in the temporal walk can be the same vertex, i.e., we can wait for an
arbitrary number of time steps at a vertex. We say that a temporal walk
\emph{spans} $x$ time steps if it starts at some time step $t$ and ends at some
time step $t'$ with $t'-t=x$. 
We assume that the lifetime $\ell$ is large enough
such that a desired temporal walk can be constructed. {For all our results
it is enough to assume $\ell \ge n^2$, as $n^2$ time steps suffice for
$\texp$~\cite{ErlebachHK21} and therefore also for $\trp$.} For any function $f:
X \rightarrow X$ for a universe $X$ we write $f^i(x)$ when we mean applying the
function $i$-times iteratively for any $x \in X$ and integer $i$, e.g.,
$f^2(x)=f(f(x))$. {We use $\circ$ to denote function composition, i.e.,}
for two functions $f: X \rightarrow X$ and $g: X \rightarrow X$, we denote with
$f \circ g$ the function that maps $x\in X$ to $f(g(x))$.

~\\ \noindent \textbf{Isomorphisms and automorphisms.} Two static graphs $G$ and $H$ are
\textit{isomorphic} exactly if a bijection $\theta: V(G) \rightarrow V(H)$
(called an \textit{isomorphism}) exists with the following property: two
vertices $u, v$ are adjacent in $G$ exactly if $\theta(u)$ is adjacent to
$\theta(v)$ in $H$. We write $G\cong H$ if $G$ is isomorphic to $H$. An
\textit{automorphism} is an isomorphism from a graph $G$ to itself. The set of
all automorphisms of a graph $G$ {forms a group $\autg$, with $\circ$ as
group operation.} Refer to~\cite{GodsilR01,KnauerK19,LauriS16} for further reading on the topic
of isomorphisms and automorphisms of graphs.

For a temporal graph $\mathcal{G}$ with lifetime $\ell$ we denote with $\auttg$
the set of all functions $\sigma$ such that $\sigma$ is an automorphism of each
graph $G_t$ at every time step $t \in [\ell]$ {(and hence also an
automorphism of the underlying graph~$G$). $\auttg$ with $\circ$ as group
operation is the automorphism group of the temporal graph~$\mathcal{G}$.}
\full{See Fig.~\ref{fig:automporphism-temporal-graph} for an illustration. }%
We remark that the automorphism group of the temporal graph~$\mathcal{G}$
is the intersection of the automorphism groups of the graphs
$G_t$ for $t\in[\ell]$ and that the automorphism group of
$\mathcal{G}$ can be viewed as the automorphism group of the underlying graph
$G$ of $\mathcal{G}$ with edge labels such that each edge is labeled with the set
of time steps in which it occurs.

\full{
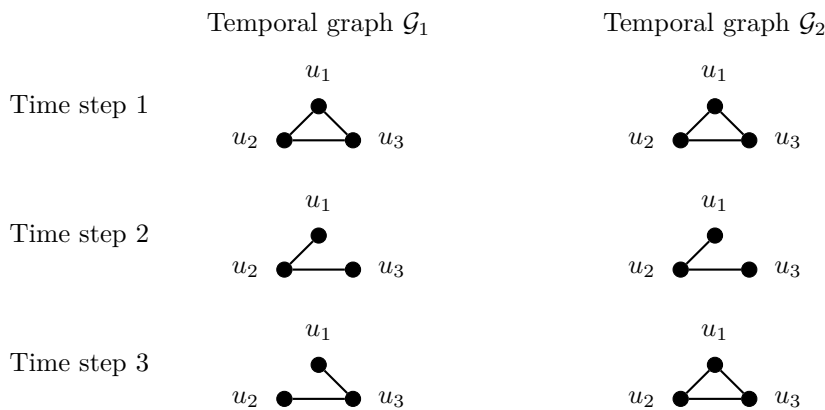
\begin{figure}[hb!]
    \centering
    \begin{tikzpicture}[vertex/.style={circle,draw,inner sep=0pt,minimum size=0.2cm}]
        
        \node[vertex,fill] (S1) at (2,0){}; 
        \node[vertex,below left=0.3cm and 0.3cm  of S1,fill] (S2){};
        \node[vertex,below right=0.3cm and 0.3 cm of S1, fill] (S3){};

        \path[draw,thick]
        (S1) edge (S2)
        (S1) edge (S3)
        (S2) edge (S3)
        ;

        \node[above=0.7cm of S1] (O1) {Temporal graph $\mathcal{G}_2$};
        \node[above=0.1cm of S1] (O1) {$u_1$};
        \node[left=0.1cm of S2] (O1) {$u_2$};
        \node[right=0.1cm of S3] (O1) {$u_3$};

        \node[vertex,left=5cm of S1,fill] (X1){}; 
        \node[vertex,below left=0.3cm and 0.3cm  of X1,fill] (X2){};
        \node[vertex,below right=0.3cm and 0.3 cm of X1, fill] (X3){};

        \path[draw,thick]
        (X1) edge (X2)
        (X1) edge (X3)
        (X2) edge (X3)
        ;

        \node[above=0.7cm of X1] (O1) {Temporal graph $\mathcal{G}_1$};
        \node[above=0.1cm of X1] (O1) {$u_1$};
        \node[left=0.1cm of X2] (O1) {$u_2$};
        \node[right=0.1cm of X3] (O1) {$u_3$};

        \node[left=2.0cm of X1] (O1) {Time step $1$};

        \node[vertex,below=1.5cm of S1, fill] (T1) {}; 
        \node[vertex,below left=0.3cm and 0.3cm  of T1,fill] (T2){};
        \node[vertex,below right=0.3cm and 0.3 cm of T1, fill] (T3){};

        \path[draw,thick]
        (T1) edge (T2)
        (T2) edge (T3)
        ;

        \node[above=0.1cm of T1] (O1) {$u_1$};
        \node[left=0.1cm of T2] (O1) {$u_2$};
        \node[right=0.1cm of T3] (O1) {$u_3$};

        \node[vertex,below=1.5cm of X1,fill] (Y1){}; 
        \node[vertex,below left=0.3cm and 0.3cm  of Y1,fill] (Y2){};
        \node[vertex,below right=0.3cm and 0.3 cm of Y1, fill] (Y3){};
        
        \path[draw,thick]
        (Y1) edge (Y2)
        (Y2) edge (Y3)
        ;

        \node[above=0.1cm of Y1] (O1) {$u_1$};
        \node[left=0.1cm of Y2] (O1) {$u_2$};
        \node[right=0.1cm of Y3] (O1) {$u_3$};

        \node[left=2.0cm of Y1] (O1) {Time step $2$};

        \node[vertex,below=1.5cm of T1, fill] (U1) {}; 
        \node[vertex,below left=0.3cm and 0.3cm  of U1,fill] (U2){};
        \node[vertex,below right=0.3cm and 0.3 cm of U1, fill] (U3){};

        \path[draw,thick]
        (U1) edge (U2)
        (U1) edge (U3)
        (U2) edge (U3)
        ;

        \node[above=0.1cm of U1] (O1) {$u_1$};
        \node[left=0.1cm of U2] (O1) {$u_2$};
        \node[right=0.1cm of U3] (O1) {$u_3$};

        \node[vertex,below=1.5cm of Y1,fill] (Z1){}; 
        \node[vertex,below left=0.3cm and 0.3cm  of Z1,fill] (Z2){};
        \node[vertex,below right=0.3cm and 0.3 cm of Z1, fill] (Z3){};
        
        \path[draw,thick]
        (Z1) edge (Z3)
        (Z2) edge (Z3)
        ;

        \node[above=0.1cm of Z1] (O1) {$u_1$};
        \node[left=0.1cm of Z2] (O1) {$u_2$};
        \node[right=0.1cm of Z3] (O1) {$u_3$};

        \node[left=2.0cm of Z1] (O1) {Time step $3$};


    \end{tikzpicture}
    \caption{Temporal graphs $\mathcal{G}_1$ and $\mathcal{G}_2$ have different
   groups $\text{Aut}(\mathcal{G}_1)$ and $\text{Aut}(\mathcal{G}_2)$.
   $\text{Aut}(\mathcal{G}_1)$ has only the identity automorphism, i.e.,
   $(u_1,u_2,u_3) \rightarrow (u_1,u_2,u_3)$ whereas $\text{Aut}(\mathcal{G}_2)$
   has the identity automorphism and additionally the automorphism
   $(u_1,u_2,u_3) \rightarrow (u_3,u_2,u_1)$. This highlights the fact that
   temporal graphs with the same underlying graphs may have different
   automorphisms since the set of automorphisms of the temporal graph is the
   \textit{intersection} of the set of automorphisms of the constituent graphs
   at each time step.} 
    \label{fig:automporphism-temporal-graph}
\end{figure}
}
The \textit{orbit} of a vertex $u$ in a temporal graph $\mathcal{G}$ with
respect to $\auttg$ is the set $V' \subseteq V$ of all vertices that $u$ can be
mapped to by automorphisms in $\auttg$. Note that, if $V'$ is the orbit of $u$,
then the orbit of every vertex in $V'$ is also~$V'$: For any two vertices $v,v'$
in $V'$, the automorphisms $\sigma$ and $\sigma'$ that map $u$ to $v$ and $v'$,
respectively, can be composed to the automorphism $\sigma' \circ \sigma^{-1}$
that maps $v$ to $v'$. Furthermore, there cannot be an automorphism $\rho$ that
maps $v$ to a vertex outside $V'$ because $\rho\circ \sigma$ would then be an
automorphism that maps $u$ to a vertex outside~$V'$; {a contradiction to
the definition of the orbit $V'$ of $u$.}
We denote by $\torb$ the set of all orbits
of the vertices of $\mathcal{G}$. Note that this set forms a partition of~$V$.
We call |$\torb$| the {\em{orbit number}} of $\mathcal{G}$. We call an edge
$\{u, v\} \in E(G_t)$ for $t \in [\ell]$ an \textit{orbit boundary edge} if $u$
and $v$ are in different orbits, and \textit{inner orbit edge} otherwise. Two
orbits connected by an orbit boundary edge {in time step~$t$} are called
{\emph{adjacent (in time step $t$)}.}

We use automorphisms to transform temporal walks, as outlined in the following.
For a temporal graph $\mathcal{G}$ with lifetime $\ell$ and any automorphism
$\sigma \in \auttg$ and any temporal walk $W=(u_t, u_{t+1}, \ldots, u_{t'})$
that starts at time $t$ and ends at time $t'$ with $t, t' \in [\ell]$, we say we
\textit{apply $\sigma$ to $W$} when we construct the temporal walk
$W'=(\sigma(u_t), \sigma(u_{t+1}), \dots, \sigma(u_{t'}))$.

~\\ \noindent \textbf{Temporal Exploration Problem.}
The \textsc{temporal exploration problem} ($\texp$) is defined for a given
temporal graph $\mathcal{G}$ with vertex set $V$ and lifetime $\ell$ and asks
for the existence of a temporal walk $W$ such that $W$ starts at a given vertex
$u \in V$ and visits all vertices of $V$. In connected temporal graphs with
sufficiently large lifetime such a walk always exists, and this is the setting
we consider throughout this work. As such, the question of existence is no
longer of interest, and instead we focus on the time span of temporal walks that
start at $u$ at time step $1$ and visit all (or a certain set of) vertices.

The following is an adaptation of a result of Erlebach et
al.~\cite{ErlebachHK21}, slightly modified to fit our notation and use case.
Intuitively speaking, the lemma states that with every extra time step at least
one additional vertex becomes reachable. Thus, starting from any vertex at any
time step, we can reach any particular other vertex in $n-1$ steps.
Consequently, there is always a temporal walk that visits all vertices of
$\mathcal{G}$ within $n^2$ time steps. 

\begin{lemma}[Reachability,~\cite{ErlebachHK21}]\label{lem:reachability} Let
    $\mathcal{G}$ be a connected temporal graph with vertex set $V$ and lifetime
    $\ell$. Denote by $R_{t, t'}(u)$ the set of vertices reachable by some
    temporal walk starting at vertex $u \in V$ at time step $t \in [\ell]$ and
    ending at time step $t' \in [\ell]$ with $t'\ge t$. If $R_{t,t'}(u)\neq V$
    and $t'<\ell$, then $R_{t,t'}(u) \subsetneq R_{t, t'+1}(u)$.
\end{lemma}

\noindent \textbf{Temporal Rendezvous Problem.} We consider a problem we call
\textsc{temporal rendezvous problem} ($\trp$) defined as follows. Let
$\mathcal{G}$ be a temporal graph with vertex set $V$ and lifetime $\ell$. Two
agents $a_1, a_2$ are placed at arbitrary vertices $u_1, u_2 \in V$
(respectively) in $G_1$ of $\mathcal{G}$. They compute respective temporal walks
$W_{1}$ and $W_{2}$ such that the agents meet at the same vertex in the same
time step at least once during these walks. The agents have the full information
of $\mathcal{G}$ available, but they cannot communicate with each other.
Furthermore, they do not know the location of the other agent, and the vertex
labels that one agent sees can be arbitrarily different from the vertex labels
that the other agent sees. The lack of consistent labels prohibits trivial
solutions such as the two agents agreeing to meet at a vertex with a specific
label, e.g., the lowest labeled vertex. We call such agents
\textit{label-oblivious}. A solution to $\trp$ is a pair of possibly different
programs $p_1$ and $p_2$ for agents $a_1$ and $a_2$, respectively, that the
agents use to compute and execute their temporal walks. The respective start
positions $u_1, u_2 \in V$ of the agents {(and the vertex labels that each
agent sees)} are chosen by an all-knowing adversary once a solution $(p_1, p_2)$
is provided. {As we assume that $\ell\ge n^2$, there is always a solution
that ensures that the agents meet: Agent $a_1$ simply waits at its start vertex,
while agent $a_2$ explores the whole graph and visits every vertex, which is
always possible within $n^2$ time steps as mentioned above. Therefore, we are
interested in solutions that enable the agents to meet as early as possible and
aim to obtain worst-case bounds significantly better than~$n^2$ on the number of
steps that are required for $\trp$.}

\section{Automorphism Utilities} 
\label{sec:aut-util}
We now introduce some helpful utilities regarding automorphisms that we use in
the following sections. Intuitively, we build up to a framework that allows us
to transform a temporal walk $W$ into a temporal walk
$W'$ that {visits more vertices that are desirable (with respect to the
exploration goal) than~$W$ does}, by applying a well-chosen automorphism to~$W$.
The following is needed to give specific guarantees that the transformed
temporal walks must uphold. {Throughout this section we use
and extend
techniques of the field of algebraic graph theory~\cite{GodsilR01,KnauerK19,LauriS16}, adapted to our specific
use cases.}

For this, we begin with some definitions. Let $\mathcal{G}$ be a temporal graph
with vertex set $V$ and let $S \in \torb$ be any orbit. For any $X \subseteq S$
and any $u \in S$ denote with $\auttg[u, X]$ the set of all automorphisms
$\sigma \in \text{Aut}(\mathcal{G})$ that map $u$ to any vertex $v \in X$. We
use the shorthand $\auttg[u,x]$ for $\auttg[u,\{x\}]$. We can then show the
following. 

\begin{lemma}\label{lem:numautos} 
    Let $S \in \torb$ be any orbit. Then
    $|\auttg[u,{x_1}]|=|\auttg[u,{x_2}]|$ for any $u, x_1, x_2 \in
    S$.
\end{lemma}
\newcommand{\proofnumautos}{
\begin{proof}
    We show the claim via proof by contradiction. Without loss of generality
    assume that $|\auttg[u,{x_1}]| >|\auttg[u,{x_2}]|$. As $x_1, x_2$ are
    contained in the same orbit, we know an automorphism $\alpha$ exists such
    that $\alpha(x_1)=x_2$, as mentioned in the preliminaries. Denote with
    $\auttg[u,{x_1}]_{\alpha}=\{\alpha \circ \sigma| \sigma \in
    \auttg[u,{x_1}]\}$. It follows by definition that $\auttg[u,{x_1}]_{\alpha}
    \subseteq \auttg[u,{x_2}]$, as all automorphisms of
    $\auttg[u,{x_1}]_{\alpha}$ map $u$ to $x_2$. Due to all automorphisms being
    bijective functions, it holds that
    $|\auttg[u,{x_1}]_{\alpha}|=|\auttg[u,{x_1}]|$. Then we have
    $|\auttg[u,{x_1}]| \leq |\auttg[u,{x_2}]|$, contradicting our assumption.
\end{proof}
}
\full{\proofnumautos}

    Let $S \in \torb$. We now consider a special 2-dimensional
    \textit{automorphism matrix $\mathcal{M}_{u, Y, X}$} for {any $u\in S$, $Y\subseteq S$, and $X \subseteq S$.}
    {It has} columns $C_1, C_2,
    \ldots, C_{|Y|+1}$, and a row for each $\sigma \in \auttg[u,X]$.
    We refer to the row for some $\sigma$ simply as row~$\sigma$.
    {The entry in row $\sigma$ of column $C_1$ is~$\sigma(u)$.}
    {Each vertex $y\in Y$ is assigned a unique column among $C_2,
    \ldots C_{|Y|+1}$ in an arbitrary way. The entry in row $\sigma$ of the
    column to which $y$ is assigned is~$\sigma(y)$.} We now give some intuition
    about the application of an automorphism matrix. Assume that we are
    constructing a temporal walk $W$ that has already visited a set $T
    \subsetneq S$ and we want to extend it to visit the vertices of $S \setminus
    T$ (with $S$ an arbitrary orbit). To facilitate this extension we first
    construct a temporal walk $W'$ that visits at least a certain number of
    vertices of $S$, but that is not guaranteed to visit vertices of $S
    \setminus T$. This walk $W'$ cannot be used to extend $W$ in a meaningful
    way. Instead, using the automorphism matrix we show that there always exists
    an automorphism $\sigma \in \text{Aut}(\mathcal{G})$ that we can apply to
    $W'$ to obtain a temporal walk $W_{\sigma}$ that visits a guaranteed
    fraction of vertices of $S \setminus T$. We can then use $W_{\sigma}$ as our
    desired extension. {Figure~\ref{fig:paths} visualizes this concept.} In
    later sections we will show how repeated application of such extensions
    leads to temporal walks that visit all vertices of an orbit~$S$. Note that
    the value $p$ in the following lemma represents the fraction of vertices
    that we still need to visit compared to all vertices in the orbit under
    consideration.
	{While $T$ is a set of already visited vertices in the application
	sketched above, the following lemma and Corollary~\ref{cor:transform2}
	are formulated more generally for arbitrary sets~$T\subseteq S$.}

\begin{lemma}\label{lem:transform1} Let $\mathcal{G}$ be a connected temporal
    graph with lifetime $\ell$ and let $S \in \torb$ be any orbit. Let $T
    \subseteq S$ with $p=(|S|-|T|)/|S|$ and $W=(u_1, u_2, \ldots, u_x)$ a
    temporal walk starting at time step $t$ and ending at $t'$ with $t, t' \in
    [\ell]$ and $u_1 \in S$ such that $W$ visits $k$ vertices of $S$. Then there
    exists a temporal walk $W'$ starting at a vertex $u' \in S$ in time step $t$
    and ending at time step $t'$ that visits at least $p k$ vertices of $S
    \setminus T$.
\end{lemma}
\newcommand{\prooftransformone}{
\begin{proof}
    Denote with $Y$ the set of $k$ vertices of $S$ visited by $W$. Construct the
    automorphism matrix $\mathcal{M}_{u_1, Y, S}$. By definition, we know that
    each column contains all vertices of~$S$ (as the matrix contains a row for
    \emph{every} $\sigma\in\auttg$). By Lemma~\ref{lem:numautos}, each
    vertex is contained the same number of times in each column. Thus, each
    column contains a fraction $p$ of vertices of $S \setminus T$. By simple
    counting arguments there exists a row $\sigma$ of $\mathcal{M}_{u_1, Y, S}$
    that contains at least $p k$ vertices of $S \setminus T$. Then applying
    $\sigma$ to $W$ yields the temporal walk $W'$ with the required
    characteristics.
\end{proof}
}
\full{\prooftransformone}
When restricting the possible start vertices where the temporal walk $W'$ of
Lemma~\ref{lem:transform1} is allowed to start, we obtain the following
corollary. In detail, we are given a set $X \subset S$ such that the walk $W'$
is only allowed to begin at a vertex $u' \in X$. In Lemma~\ref{lem:transform1},
$W'$~was allowed to start at any vertex of~$S$. Our use case for the following
corollary is that $X$ is a set of vertices that can be reached faster, making
them better candidates for start vertices when extending walks in the way
sketched above Lemma~\ref{lem:transform1}.

\begin{corollary}\label{cor:transform2} Let $\mathcal{G}$ be a connected
    temporal graph with lifetime $\ell$ and let $S \in \torb$ be any orbit. Let
    $T \subsetneq S$ and $W$ a temporal walk starting at time step $t$ and
    ending at $t'$ with $t, t' \in [\ell]$ such that the first vertex of $W$ is
    in $S$ and such that $W$ visits $k$ different vertices of $S$. For any $X
    \subseteq S$ with $|X| > |T|$ there exists a temporal walk $W'$ starting at
    a vertex $u' \in X$ at time step $t$ and ending at time step $t'$ that
    visits at least $(c-1)/c \cdot k$ vertices of $S \setminus T$, with
    $c=|X|/|T|$.
\end{corollary}
\newcommand{\prooftransformtwo}{
\begin{proof}
    We proceed along the lines of the proof of Lemma~\ref{lem:transform1} with
    $Y$ the set of vertices of $S$ visited by $W$, but instead of the matrix
    $\mathcal{M}_{u, Y, S}$ we construct the matrix $\mathcal{M}_{u, Y, X}$. By
    construction, the first column of $\mathcal{M}_{u, Y, X}$ contains all
    vertices of $X$ {and each vertex occurs the same number $m$ of times in
    the first column. In all other columns, we also have $m|X|$ vertices (not
    all may be different, but each vertex occurs at most $m$ times). In the
    worst case, $m |T|$ vertices are from $T$---to see this, note that every
    vertex of $V$ occurs $m$ times in every column in $\mathcal{M}_{u, Y, S}$
    and in $\mathcal{M}_{u, Y, X}$ no new occurrences of a vertex are
    introduced.} Thus, each column contains at least $m|X|-m|T|$ vertices of
    $S\setminus T$, and so the fraction of vertices from $S\setminus T$ in each
    column is at least $(|X|-|T|)/|X|=(c-1)/c$. Now, by the same counting
    arguments as those used in the proof of Lemma~\ref{lem:transform1}, we know
    that there exists a row $\sigma$ that contains at least $(c-1)/c \cdot k$
    vertices of $S \setminus T$. Therefore, we can construct the desired
    temporal walk $W'$ by applying $\sigma$ to $W$.
\end{proof}
}
\full{\prooftransformtwo}

\section{Upper Bounds for TEXP}
\label{sec:texp}
A common approach to build a temporal walk for $\texp$ is to use
Lemma~\ref{lem:reachability}, i.e., to construct a (large) set $X$ of reachable
vertices so that an unseen vertex $v$ of the current walk is in the set $X$ and
the walk can then be extended by~$v$. We are interested in exploring the
vertices of one orbit quickly, as this will be useful for $\trp$ in
Section~\ref{sec:rendezvous} where the agents try to meet in one orbit, and for
$\texp$ because we can explore a temporal graph orbit by orbit. Therefore, we
want to find walks visiting many vertices of one orbit. Our approach is similar
to the common approach mentioned above, and so we want to construct a (large)
set $X$ of reachable vertices, but now with the property that $X$ is a subset of
the orbit under consideration. To construct $X$, we show in
Lemma~\ref{lem:laneorbit} a kind of ``reachability between orbits.'' 

To describe this in more detail, we need the concept of so-called lanes.
Intuitively, lanes are defined for a set of vertices that are all contained in
some single orbit, and give us knowledge about the vertices that are quickly
reachable while only using orbit boundary edges in each time step. Using this
concept of lanes we derive a first result for exploring a single large orbit
with a temporal walk that spans $O((n^{5/3}+rn)\log n)$ time steps
(Theorem~\ref{thm:orbitexplore}). In the proof of that lemma we build the final
temporal walk iteratively, by concatenating multiple smaller temporal walks. To
make sure each new such small temporal walk visits a desired number of vertices
not yet visited, we use Lemma~\ref{lem:transform1}, which---informally---lets us
transform temporal walks that visit too many previously visited vertices into
temporal walks that visit many previously unvisited vertices.

We follow this up with a more refined technique that considers the size of the
orbit $S$ one wants to explore as a
parameter, but also uses the concept of lanes and walk transformations sketched
above. It gives us an upper bound of $O(|S|^{1+\epsilon}+n\log |S|)$, for any
constant $\epsilon > 0$. This result is formulated in
Theorem~\ref{thm:epsilonexploration}. Finally, we use a repeated application of
Theorem~\ref{thm:epsilonexploration} to achieve an upper bound for
$\textsc{TEXP}$ of $O(rn^{1+\epsilon})$. We start with an auxiliary lemma that
focuses on the orbit boundary edges between two orbits.

\begin{lemma}\label{lem:orbitedges} {Let $G_t$ be the graph at time
    step~$t$ in a connected temporal graph~$\mathcal{G}$ and $S,S' \in \torb$,
    and let $G'$ be the subgraph of $G_t$ that contains only orbit boundary
    edges. Then all vertices in $S$ have the same degree in the bipartite graph
    $G'[S \cup S']$.}
\end{lemma}
\newcommand{\prooforbitedges}{
\begin{proof} 
    First, the lemma is trivially true if the orbits $S$ and $S'$ are not
    adjacent in time step~$t$. Thus, for the rest of the proof we consider the
    situation that $S$ and $S'$ are adjacent. Assume for a contradiction that
    not all vertices of $S$ have the same degree in $G'$, and let $u, v$ be two
    arbitrary vertices of $S$ with respective degrees $d_u$ and $d_v$ in $G'$
    such that $d_u < d_v$. By definition of an orbit there exists a $\sigma \in
    \auttg$ with $\sigma(v)=u$. Since $\sigma$ must map neighbors of $v$ to
    neighbors of $u$ and since $v$ has more neighbors than $u$ in $S'$, $\sigma$
    must map at least one vertex $w$ of $v$ to a neighbor $w'$ of $u$ outside of
    $S'$. However $\sigma(w)=w'$ implies that $w'\in S'$, a contradiction.
\end{proof}
}
\full{\prooforbitedges}

\begin{figure}[!t]
	\centering
	\begin{subfigure}[t]{0.45\textwidth}
		\centering
		\includegraphics{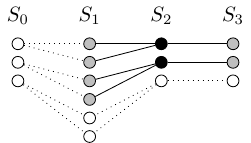}
		\caption{}
		\label{fig:lanes1}
	\end{subfigure}
	\begin{subfigure}[t]{0.45\textwidth}
		\centering
		\includegraphics{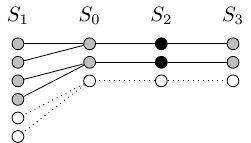}
		\caption{}
		\label{fig:lanes2}
	\end{subfigure}
	\begin{subfigure}[t]{0.45\textwidth}
		\centering
		\includegraphics{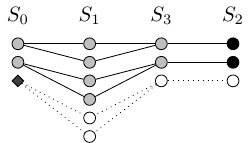}
		\caption{}
		\label{fig:lanes3}
	\end{subfigure}
	\begin{subfigure}[t]{0.45\textwidth}
		\centering
		\includegraphics{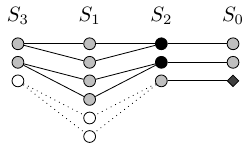}
		\caption{}
		\label{fig:lanes4}
	\end{subfigure}
    \caption{Visualization of the properties of a lane $L_{t, t+r}(X)$.
    Each vertex in the set $X$ is colored black, and the set of reachable vertices in each
    time step is colored gray. Each figure represents one additional time step.
    In the second time step (Figure~b) all vertices of the lane are reachable.
    In the third time step (Figure~c) one vertex outside the lane must be
    reachable, for example the diamond shaped vertex. This is due to the fact
    that the temporal graph at hand is connected, and thus at least one additional
    vertex is reachable with every next time step. Here, one additional
    time step then suffices to reach a vertex of $S_2 \setminus X$ (Figure~d).}
    \label{fig:lanes}
\end{figure}

To describe  
reachability between orbits, we have to introduce some extra notation. Let
$\mathcal{G}$ be a temporal graph with lifetime $\ell$ and vertex set $V$ and $S
\in \torb$ be an orbit. We call a \textit{lane} $L_{t, t'}(X)$ with $X \subseteq
S$ and $t, t' \in [\ell]$ the set of all vertices reachable from any $u \in X$
in $G$ by any temporal walk $W$ that only uses orbit boundary edges and starts
in time step $t$ and ends in time step at most $t'$. We write $L_{t, t'}(u)$
instead of $L_{t, t'}(\{u\})$. {See Fig.~\ref{fig:lanes} for some intuition.}

The next lemma gives us a lower bound on the number of vertices of an orbit $S'$
that can be reached from a subset $X$ of the vertices of an orbit $S$
{within $r$ time steps}. Intuitively speaking, we show a lower bound on the
number of vertices {\em reachable from orbit $S$ {in} another orbit $S'$}.
{A simple consequence of the following lemma is that, from any start vertex
in the temporal graph, at least one vertex in every orbit is reachable within
$r$ time steps.}

\begin{lemma}[Reachability between Orbits]\label{lem:laneorbit} Let
    $\mathcal{G}$ be a connected temporal graph with lifetime $\ell$ and $S \in
    \torb$. For any $X \subseteq S$ and $S' \in \torb$ it holds that $|L_{t,
    t'}(X) \cap S'| \geq  \lceil |X| \cdot |S'|/|S| \rceil$ for any $t \in
    [\ell]$ and $t'=t+r$, where $r=|\torb|$ is the orbit number.
\end{lemma}
\newcommand{\prooflaneorbit}{
\begin{proof}
    We say an orbit $S^*$ is \textit{processed for $r'$} if $|L_{t, t+r'}(X)
    \cap S^*| \geq \lceil |X| \cdot |S^*|/|S|\rceil$ for some $r' \leq r$, and
    \textit{unprocessed for $r'$} otherwise. 
    We now show that as long as unprocessed orbits remain, there exists an orbit
    $S'$ such that $S'$ is unprocessed for $r'$, but processed for $r'+1$, i.e.,
    in each time step {at least} one unprocessed orbit becomes processed.
    The lemma then follows simply by counting the number of orbits.

    For $r'=0$, orbit $S$ is trivially processed for~$r'$. For $r'=1$ note that
    there exists an orbit $S'$ that is adjacent to $S$ in time step~$t$.
    Lemma~\ref{lem:orbitedges} shows that in $G'=G'_{t}[S \cup S']$ all vertices
    of $S$ have the same degree $d_{S}$, and all vertices of $S'$ have the same
    degree $d_{S'}$, where $G'_{t}$ is the subgraph of $G_{t}$ containing only
    the orbit boundary edges that are present in time step~$t$. This implies
    $|E(G')|=|S|\cdot d_S = |S'|\cdot d_{S'}$ and hence $d_S/d_{S'}=|S'|/|S|$.
    Denote by $E_X$ the set of edges incident to vertices of $X$ in $G'$. It
    holds that $|E_X| = |X| \cdot d_{S}$. By definition, a vertex $u' \in S'$ is
    incident to at most $d_{S'}$ edges of $E_X$. Denote with $X' \subseteq S'$
    the set of vertices of $S'$ that are neighbors of vertices of $X$ in $G'$.
    It holds that $|X'| \geq \lceil |X| \cdot d_{S}/d_{S'} \rceil=\lceil |X|
    \cdot |S'|/|S| \rceil$. As $L_{t,t+1}(X)\cap S'\supseteq X'$, this means
    that $S'$ is processed for $r'=1$. {See Fig.~\ref{fig:lanes1} for a
    visualization of this.}

    For $r'$ with $1\le r'< r$, if there is still an unprocessed orbit for~$r'$,
    there must exist an unprocessed orbit $S_u$ adjacent to a processed orbit
    $S_p$ in time step $t+r'$. This follows from the fact that $\mathcal{G}$ is
    connected. As $S_p$ is processed for~$r'$, we know that the set
    $X_p=L_{t,t'}(X)\cap S_p$ satisfies $|X_p| \ge \lceil |X| \cdot |S_p|/|S|
    \rceil$. By the same argument as in the previous paragraph, it follows that
    the set $X_u\subseteq S_u$ of vertices in $S_u$ that are neighbors of
    vertices in $S_p$ in time step $t+r'$ satisfies $|X_u|\ge \lceil |X_p| \cdot
    |S_u|/|S_p|\rceil \geq \big\lceil \lceil |X| \cdot |S_p|/|S| \rceil \cdot
    |S_u|/|S_p| \big\rceil \geq \lceil |X| \cdot |S_u|/|S| \rceil$. As
    $L_{t,t+r'+1}(X)\cap S_u\supseteq X_u$, it follows that $S_u$ is processed
    for $r'+1$. As at least one orbit becomes processed in each time step, all
    orbits are processed latest for $r'=r$.
\end{proof}
}
\full{\prooflaneorbit}

By using Lemma~\ref{lem:reachability} and Lemma~\ref{lem:laneorbit}, we now
bound the number of time steps needed to reach a set of $h$ vertices within an
orbit $S$.

\begin{lemma}\label{lem:nextreachable} Let $\mathcal{G}$ be a connected temporal
graph with lifetime $\ell$ and vertex set $V$. Let $S \in \torb$ and let
$r=|\torb|$ be the orbit number. For any $h \leq |S|$, start vertex {$u \in S$}
{and start time $t$}, there exists a set $X \subset S$ {with $|X|=h$} such that
we can reach any vertex in $X$ in at most $O(\min\{h \cdot n / |S|, h r\}+r)$
time steps. That is, for every vertex $u'$ of $X$, we have a temporal walk
starting at $u$ at time step $t$ and ending at $u'$ at time step $t'$ with
$t'-t=O(\min\{h \cdot n / |S|, h r\}+r)$ {such that $<h$ vertices of $S$ are reachable
by time step $t'-1$.}
\end{lemma}

\newcommand{\proofnextreachable}{
\begin{proof}
    We first show how to reach $h$ vertices in $O(h \cdot n/|S| + r)$ time
    steps. Denote with $R_{t, t'}(u)$ the set of vertices reachable with a
    temporal walk from $u$ at time step $t$ and ending at or before time step
    $t'$. Our upper bound on the number of time steps required is based on first
    expanding the set of reachable vertices until we achieve an
    \textit{overflow} in some orbit $S' \in \torb$, which we define as $|S' \cap
    R_{t, t'}(u)| \geq \texttt{min}\{|S'|, \lceil h\cdot |S'|/|S|\rceil+1\}$. As
    $\mathcal{G}$ is connected, we know that, as long as $R_{t, t'}(u)\neq V$,
    $R_{t, t'}(u) \subsetneq R_{t, t'+1}(u)$ (Lemma~\ref{lem:reachability}).
    From this it follows that an overflow exists after at most $\sum_{S' \in
    \torb} (\lceil h \cdot |S'|/|S|\rceil + 1)\leq \lceil h
    n/|S|\rceil+{2}r$ time steps, due to the pigeonhole principle. {See
    Fig.~\ref{fig:lanes3} for an example of an overflow.} Let $t'$ be the time
    step when the first overflow occurred, and let $S'$ be an orbit where an
    overflow occurred in time step~$t'$. Denote by $X'$ the set of vertices
    reachable in $S'$ by time step~$t'$. Using Lemma~\ref{lem:laneorbit} with
    lane $L_{t', t'+r}(X')$ we know that once an overflow occurs in $S'$, after
    $r$ additional time steps an overflow occurs in all orbits, which guarantees
    that $h$ vertices in $S$ are reachable: {If $R_{t,t'}(u)\cap S'=S'$,
    Lemma~\ref{lem:laneorbit} implies $|L_{t',t'+r}\cap S|\ge \lceil |S'|\cdot
    |S|/|S'|\rceil=|S|$, meaning all vertices of $S$ are reachable by time
    $t'+r$. If $|R_{t,t'}(u)\cap S'|\ge \lceil h \cdot |S'|/|S|\rceil +1$,
    Lemma~\ref{lem:laneorbit} implies $|L_{t',t'+r}\cap S|\ge \lceil (\lceil h
    \cdot |S'|/|S|\rceil +1)\cdot |S|/|S'|\rceil\ge h$, meaning that at least
    $h$ vertices of $S$ are reachable by time $t'+r$.} In total, $O(h \cdot n /
    |S| + r)$ time steps always suffice to guarantee an overflow in orbit $S$,
    which means we can reach $h$ vertices of $S$.

    Next, we show how to reach $h$ vertices of $S$ in $O(hr)$ time steps.
        Let $X$ be the set of reachable vertices
        of $S$. Initially (in time step $t$) $X=\{u\}$. After $r-1$ time steps
        all vertices in lane $L_{t, t+r}(X)$ are reachable. By
        Lemma~\ref{lem:laneorbit}, $|L_{t, t+r}(X)\cap S'|\geq \lceil |X| \cdot
        |S'|/|S| \rceil$ for any $S' \in \torb$. 
        After one additional time step (Lemma~\ref{lem:reachability}) there
        exists some orbit $S'$ such that at least one more vertex $v \in
        S'\setminus L_{t, t+r}(X)$ is reachable. {See
        Fig.~\ref{fig:lanes3} for an example of this.} So the set $X' :=
        (L_{t, t+r}(X) \cap S')\cup \{v\}$ is reachable after $r$ time steps.
        Note that $|X'| \geq \lceil |X| \cdot |S'|/|S| \rceil + 1$. Using
        Lemma~\ref{lem:laneorbit} again, now for start time $t+r+1$ and
        initially reachable vertex set $X'$ in orbit~$S'$, we get that the set
        $X'':= L_{t+r+1, t+2r+1}(X') \cap S$ satisfies $|X''| \geq \lceil |X'|
        \cdot |S|/|S'| \rceil \geq \lceil( \lceil|X| \cdot |S'|/|S|\rceil+1)
        \cdot |S|/|S'|\rceil \geq \lceil |X| + |S|/|S'| \rceil \ge |X|+1$.
        Therefore, we can reach a vertex $w$ in $S \setminus X$ by time
        $t+2r+1$. Repeating the construction with $X=\{u, w\}$ and then with
        sets $X$ of size $3, \ldots, h-1$, after $O(rh)$ time steps there are
        $h$ vertices of $S$ reachable.
\end{proof}
}
\full{\proofnextreachable}

    Next we present Theorem~\ref{thm:orbitexplore}, which states an upper bound
    for visiting all vertices of a given orbit $S$. The rough idea used in the
    proof is that we iteratively build the final temporal walk $W$ by
    concatenating smaller temporal walks. In each step of the iteration, a small
    temporal walk $W'$ is first constructed via Lemma~\ref{lem:nextreachable} to
    visit a subset of the vertices of $S$, which are not necessarily unvisited,
    but such that the size of the subset is at least a certain threshold value.
    Using Lemma~\ref{lem:transform1} we find an automorphism $\sigma$ that we
    apply to $W'$ to obtain a temporal walk $W_{\sigma}$ that visits many
    unvisited vertices of $S$. We then extend $W$ via this transformed walk
    $W_{\sigma}$\conf{ (see Fig.~\ref{fig:paths} for a sketch of the proof
    idea)}. In this way we can explore all vertices of a large orbit $S$ faster
    than by repeated application of Lemma~\ref{lem:reachability}. The key to
    obtain a good bound on the number of time steps required is to find a good
    value for the number of vertices of $S$ visited by each small temporal walk.

\begin{theorem}\label{thm:orbitexplore} Let $\mathcal{G}$ be a temporal graph
    with lifetime $\ell$ and vertex set $V$. Take $S \in \torb$ and $r=|\torb|$
    the orbit number. For any $t\in[\ell]$ there exists a temporal walk $W$
    starting at time step $t$ that visits all vertices of $S$ and ends at time
    step $t'$ with $t'-t=O((n^{5/3}+rn)\log n)$. 
\end{theorem}
\newcommand{\prooforbitexplore}{
\begin{proof}
    In the following we assume that $|S| = \Omega(n^{2/3})$ as otherwise we can
    visit all vertices of $S$ in $O(n^{5/3})$ time steps via
    Lemma~\ref{lem:reachability}. 

    We build $W$ iteratively. Initially, $W$ is an empty temporal {walk.
    Denote by} $T \subseteq S$ the set of vertices visited by $W$ and by $t_W$
    the current last time step in which $W$ visits a vertex. First we describe a
    subroutine that yields a temporal walk {$W'$} that visits $\lceil
    n^{1/3}\rceil$ vertices of $S$, but without guarantee that the visited
    vertices are from $S \setminus T$.
    
    Choose $u\in S$ arbitrarily and let $W'$ start at $u$ in time step $t_W+n$.
    Denote with $X$ the set of vertices of $S$ that we have visited so far (initially
    $X = \{u\}$) during the construction of $W'$. Use
    Lemma~\ref{lem:nextreachable} to extend $W'$ by a vertex $u' \in S \setminus
    X$, set $X:=X\cup \{u'\}$, and repeat iteratively until $|X|=\lceil
    n^{1/3}\rceil$. Each application of Lemma~\ref{lem:nextreachable} yields a
    temporal walk that visits a vertex of $S\setminus X$ in $O((|X|+1) n/|S|+r)$
    time steps. Applying Lemma~\ref{lem:nextreachable} $\lceil
    n^{1/3}\rceil$-times in this fashion yields a temporal walk $W'$ that visits
    $\lceil n^{1/3}\rceil$ vertices of $S$ within $\sum_{i=0}^{\lceil
    n^{1/3}\rceil} O(i \cdot n/|S| + r)=O(n^{2/3} \cdot n/|S| + r n^{1/3})$ time
    steps, and as $|S|=\Omega(n^{2/3})$ this is bounded by $O(n +rn^{1/3})$.

    By Lemma~\ref{lem:transform1} there exists a temporal walk {$W_{\sigma}$} that
    starts at time step $t_W+n$ and visits at least $p n^{1/3}$ vertices of $S
    \setminus T$ with some $v \in S$ the first vertex of $W_{\sigma}$ and
    $p=(|S|-|T|)/|S|$, which can be obtained by first constructing $W'$ as
    outlined in the previous paragraph and then applying some automorphism
    $\sigma \in \auttg$ to~$W'$. By Lemma~\ref{lem:reachability} we also know
    that there exists a temporal walk $W_{v}$ that starts {at the vertex where $W$
    ends at time step $t_W$ and after $n$ time steps ends with the first vertex $v$ of
    $W_{\sigma}$.} We extend $W$ by $W_{v}$ and $W_{\sigma}$. See
    Fig.~\ref{fig:paths} for a sketch of this idea.
    
    We call one such extension of $W$ a \textit{phase} and the
    number of vertices we add to $T$ due to a phase the \textit{progress}. It is
    easy to see that as long as $p \geq \frac{1}{2}$ each phase yields at least
    $\frac{1}{2} n^{1/3}$ progress. It follows that after $n^{2/3}$ phases the
    fraction $p$ is less than $\frac{1}{2}$. Now, as long as $p \geq
    \frac{1}{4}$ each phase yields at least $\frac{1}{4} n^{1/3}$ progress, and
    after $n^{2/3}$ additional phases $p < \frac{1}{4}$. One can see that after
    $O(n^{2/3} \log n)$ phases we have visited all vertices of~$S$. Each phase
    yields an extension of $W$ that spans $O(n+rn^{1/3})$ time steps, i.e., we
    have fully constructed $W$ after $O((n^{5/3}+rn)\log n)$ time steps. 
\end{proof}
}
\full{\prooforbitexplore}
\begin{figure}
    \centering
    \includegraphics[width=\textwidth]{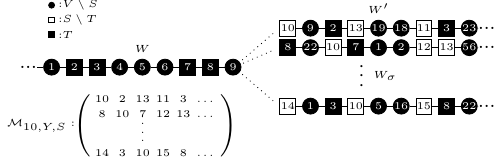}
    \caption{The construction scheme of Theorem~\ref{thm:orbitexplore}. $W$ is
    the temporal walk constructed so far. We aim to extend the walk with a walk
    $W_{\sigma}$ that visits many vertices of $S\setminus T$, where $T$ is the set of
    vertices of orbit $S$ that we have already visited. To find $W_{\sigma}$, we
    construct $W'$ and the automorphism matrix (bottom left)
    for the vertex with label $10$ (the start vertex of $W'$), the set $Y$ of vertices of $S$
    visited by $W'$, and the entire orbit $S$ as the set of possible start vertices of~$W_{\sigma}$.
    One of the rows in the matrix then gives us an automorphism $\sigma$ that, when applied to $W'$,
    yields the desired walk~$W_{\sigma}$.}
    \label{fig:paths}
\end{figure}

    {The following lemma is concerned with visiting a fraction $1/c$
of the vertices of a given orbit $S$ with a temporal walk. One significant
contribution to the number of time steps required by the temporal walk
constructed in Theorem~\ref{thm:orbitexplore} is the use of
Lemma~\ref{lem:transform1}. Roughly speaking, Lemma~\ref{lem:transform1}
provides a temporal walk that visits a large number of unvisited vertices, but
with the caveat that every vertex of $S$ can potentially be the start vertex of this
transformed walk (instead of restricting the potential start vertices
for the walk to a smaller subset, which might be reachable more quickly).
The consequence of this is that for each such transformation we require,
we must plan a ``buffer'' of $n$ time steps to ensure that all vertices of $S$ are
reachable by the time step in which the transformed walk starts (Lemma~\ref{lem:reachability}).
Corollary~\ref{cor:transform2} provides a ``trade-off'' for this: a decrease in
the set of possible start vertices of the transformed walk $W_{\sigma}$
decreases the number of previously unvisited vertices $W_{\sigma}$ visits, but
also decreases the number of time steps required to reach the first vertex of
$W_{\sigma}$. Using this property we construct a recursive algorithm that visits
a fraction of the vertices of $S$ quickly instead of applying the iterative construction of
Theorem~\ref{thm:orbitexplore}. In our recursive construction,
the walks we concatenate shrink with each recursive call. If we were to use
Lemma~\ref{lem:transform1} during this, we would have an additional $n$ time
steps with each recursive call. Instead, Corollary~\ref{cor:transform2} lets us
reduce the number of possible start vertices dramatically. The time span required
by this walk is then not dependent on $n$, but dependent on $|S|$ and $r$ (the
orbit number), and thus is especially useful for exploring smaller orbits.
\full{Later in this section we show how this }\conf{This }can
then be used iteratively to construct a
temporal walk that visits all vertices of $S$, which we in turn use to visit all
vertices $V$ by visiting all orbits one after the other. }

\newcommand{\ww}{u}

\begin{lemma}\label{lem:vtexplorationub} Let $\mathcal{G}$ be a connected
    temporal graph with lifetime $\ell$ and vertex set $V$. Let $S \in \torb$
    and let $r=|\torb|$ be the orbit number. For any $t\in[\ell]$ and any
    {$\ww \in S$} there exists a temporal walk $W$ that starts at vertex
    $\ww$ in time step $t$ and visits a fraction $1/c$ (for any $1 < c <\
    |S|$) of the vertices of $S$ such that $W$ spans $O(rc(|S|/c)^{\phi(c)}
    \log |S|)$ time steps, with $\phi(c)=1/(\log {f(c)})$ and
    $f(c)=(1+(c-1)/c)$.

\end{lemma}
\newcommand{\proofvtexplorationub}{
\begin{proof}
   We use a recursive construction based on the following idea: If we
   can recursively construct temporal walks that visit $k$ vertices
   of~$S$, then we can construct two such walks $W_1$ and $W_2$,
   transform $W_2$ into a walk $W_{\sigma}$ that visits $k(c-1)/c$ vertices of $S$
   that have not been visited by $W_1$, and combine the two
   walks to get a walk that visits $f(c)\cdot k$ vertices of~$S$.
   By allowing only $ck$ potential start vertices for $W_{\sigma}$,
   $O(rck)$~time steps suffice for reaching the start vertex
   of $W_{\sigma}$ from the end vertex of~$W_1$.

   In more detail,
   assume that, for some $k$ with $1 \leq k \leq |S|/c$, $W_{1}$
   is a temporal walk starting at a vertex of $S$ at some time step $t'$,
   ending at some time step $t''$, and visiting $k$ different vertices of $S$.
   Denote the set of these $k$ different vertices with $T\subseteq S$. Let
   {$W_{2}$} be an arbitrary temporal walk starting at time step
   {$t''+O(rck)$} that also visits $k$ different vertices of $S$, and assume that
   {$W_{2}$} starts at a vertex of $S$.
   By Corollary~\ref{cor:transform2}, we know that for any subset $X$ of $S$ of size
   $ck$ there exists a {temporal walk $W_{\sigma}$}, obtained by
   applying an automorphism $\sigma$ to $W_2$, that starts at some $x \in
   X$ and visits $((c-1)/c)k$ vertices of $S \setminus T$.

    For any integer $x$, we denote with $T(x)$ the number of time steps our
    construction scheme requires for a temporal walk that visits $x$ different
    vertices of~$S$, assuming that we start on a vertex of $S$. Then, $T(k)$
    time steps suffice to find the initial temporal walk $W_{1}$ as
    well as $W_{2}$. Lemma~\ref{lem:nextreachable}, applied with $h=ck$,
    lets us reach $ck$ vertices of $S$ within $O(ckr)$ time steps
    (note that $ck\le|S|$ since we assume $k\le |S|/c$). Using those $ck$ vertices
    as possible start vertices in the application of
    Corollary~\ref{cor:transform2}, we can therefore reach the start vertex
    of $W_{\sigma}$ in $O(ckr)$ time steps.

    To construct $W$, we first recursively construct a temporal walk $W_{1}$
    starting from $\ww$ and visiting $k$ vertices of $S$ in $T(k)$ time steps.
    Afterwards we use $O(rck)$ time steps to construct a set $S^{\mathrm{R}}$ of
    $ck$ vertices of $S$ that are reachable from the last vertex of
    $W_{1}$. Taking an arbitrary vertex of $S^{\mathrm{R}}$ as start
    vertex, we construct a walk {$W_{2}$} visiting $k$ vertices
    of $S$ in $T(k)$ time steps. By using Corollary~\ref{cor:transform2},
    we can turn {$W_{2}$} into a walk {$W_{\sigma}$} visiting
    $(c-1)/c \cdot k$ vertices of $S$ that are not yet visited by
    $W_{1}$. Combining $W_{1}$ with a walk from the last
    vertex of $W_{1}$ to the start vertex of
    {$W_{\sigma}$} and {finally} with {$W_{\sigma}$} we get a walk
    $W$ that visits $k+(c-1)/c \cdot k=f(c)k$ vertices of $S$ in $2T(k)+O(rck)$
    time steps. Overall, we obtain the recurrence $T(f(c)k)
    = 2T(k)+O(rck)$, which represents the number of time steps used to construct a
    walk $W$ visiting $f(c)k$ different vertices of $S$ if we start at a vertex
    of $S$. If we start to expand the recursive calls of the function as shown
    below (where $\beta$ is the constant hidden in the $O(rck)$ term), we can
    rewrite the function as a sum:
    \begin{align*} 
        T(k)    &= \beta rck/f(c) + 2T(k/f(c))\\ 
                &=  \beta rck/f(c) + 2\beta rck/f^2(c) + 4T(k/f^2(c))\\ 
                &=  \beta rck/f(c) + 2\beta rck/f^2(c) + 4\beta rck/f^3(c) + 8T(k/f^3(c))\\ 
                &= \ldots
    \end{align*}
    Recall that $\phi(c)=1/(\log {f(c)})$. After $\phi(c) \log k=\log_{f(c)} k$ recursive
    calls (i.e., at recursion depth $\phi(c)\log k$), it is guaranteed that the argument to $T(\cdot)$ is of size at
    most~1. Due to the number of recursive calls we have at most $2^{\phi(c)\log k}=
    k^{\phi(c)}$
    leaf nodes in the recursion tree, with each such node contributing $T(1)=1$
    to the total number of required time steps. We then get:
    \begin{align*} 
        T(k)    &= \beta rck \sum_{i=1}^{\phi(c)\log k} \frac{1}{2}(\frac{2}{f(c)})^{i} + k^{\phi(c)}\\ 
                &\leq \frac{\beta rck}{2} (\frac{2}{f(c)})^{\phi(c) \log k}\phi(c)\log k + k^{\phi(c)}\\ 
                &\leq \frac{\beta rck}{2} (\frac{2}{f(c)})^{\phi(c) (\log_{\frac{2}{f(c)}} k) / (\log_{\frac{2}{f(c)}} 2)}\phi(c)\log k + k^{\phi(c)}\\ 
                &= O(rck^{1+\phi(c)/(\log_{\frac{2}{f(c)}} 2)} \log k + k^{\phi(c)})
                = O(rck^{1+\phi(c)\cdot(\log \frac{2}{f(c)})} \log k + k^{\phi(c)})\\
                &= O(rck^{1+\phi(c)\cdot(1 - \log {f(c)})} \log k + k^{\phi(c)}) 
                = O(rck^{1+\phi(c)\cdot(1 - \frac{1}{\phi(c)})} \log k + k^{\phi(c)}) \\
                &= O(rck^{1+\phi(c) - 1} \log k + k^{\phi(c)})
                = O(rck^{\phi(c)} \log k )
    \end{align*}
    Thus, we can visit $|S|/c$ vertices of $S$ in
    $T(|S|/c)=O(rc(|S|/c)^{\phi(c)} \log |S|) $ time steps. 
\end{proof}
}
\full{\proofvtexplorationub}
\begin{corollary}\label{cor:tinyfraction} Let $\mathcal{G}$ be a temporal graph
    with lifetime $\ell$ and vertex set $V$. Let $S \in \torb$ and $r=|\torb|$
    be the orbit number. For any $t\in[\ell]$, any $u \in S$, and any fixed
    $\epsilon > 0$,  there exists a temporal walk $W$ that starts at vertex $u$
    in time step $t$ and visits some constant fraction $\alpha < 1$  
    of the vertices of
    $S$ such that 
    $W$ spans $O(r|S|^{1+\epsilon})$ time steps.     
\end{corollary}
\newcommand{\prooftinyfraction}{
\begin{proof}
Take $f(c)=((c-1)/c)+1)$ and $\phi(c)=1/(\log {f(c)})$ from
Lemma~\ref{lem:vtexplorationub}. Note that $\phi(c)$ converges to $1$ from above for $c$
going to infinity. Therefore, for any given constant $\epsilon>0$,
there exists a constant $c$ such
that $\phi(c)\le 1+0.9\epsilon$.
Applying Lemma~\ref{lem:vtexplorationub} with that value of $c$
yields a temporal walk with
$O(rc(|S|/c)^{\phi(c)} \log |S|)
=O(r|S|^{1+0.9\epsilon}\log|S|)=O(r|S|^{1+\epsilon})$ time steps
that visits a $\frac{1}{c}$ fraction of the vertices in~$S$.
We can thus choose $\alpha=\frac{1}{c}$, and the proof is complete.
\end{proof}
}
\full{\prooftinyfraction}
We now present an improved version of Theorem~\ref{thm:orbitexplore} 
for exploring a whole orbit.

\begin{theorem}\label{thm:epsilonexploration} Let $\mathcal{G}$ be a temporal
    graph with lifetime $\ell$ and vertex set $V$. Let $S \in \torb$ and
    $r=|\torb|$ be the orbit number. For any $t\in[\ell]$, any $u \in V$, and
    any fixed $\epsilon > 0$, there exists a temporal walk $W$ that starts at
    vertex $u$ in time step $t$ and visits all vertices of $S$ such that $W$
    spans $O(r|S|^{1+\epsilon} + n \log |S|)$ time steps.
\end{theorem}
\newcommand{\proofepsilonexploration}{  
\begin{proof}
    The proof is similar to the construction used in
    Lemma~\ref{lem:vtexplorationub}. We build the final temporal walk
    iteratively. Initially, $W$ is an empty temporal walk starting
    at $u$ at time step~$t$. Denote by $T$ the set
    of vertices visited by $W$ and by $t_W$ the current last time step of $W$.

    Corollary~\ref{cor:tinyfraction}, applied with $\epsilon'=0.9\epsilon$
    as the value for $\epsilon$ in the statement of that corollary, can be used to construct a temporal walk
    $W_1$ of $O(r|S|^{1+0.9\epsilon})$ time steps that visits $\alpha |S|$ vertices
    of $S$, for some constant $\alpha>0$, and starts at time step $t_W+n$. Using
    Lemma~\ref{lem:transform1} we can transform $W_1$ into a walk $W_{\sigma}$
    that visits $p \cdot \alpha |S|$ vertices of $S\setminus T$, with
    $p=(|S|-|T|)/|S|$. To reach the first vertex of $W_{\sigma}$ in time step
    $t_W+n$ we use standard reachability (Lemma~\ref{lem:reachability}).

    We call one such extension of $W$ a \textit{phase} and the number of
    vertices we add to $T$ due to a phase the \textit{progress}. As long as $p
    \ge 1/2$ each phase yields at least $\frac12 \alpha |S|$ progress. It follows
    that after $1/\alpha$ phases we have visited $|S|/2$ vertices of $S$. Now, as
    long as $p \ge 1/4$ we make $\frac14\alpha|S|$ progress per phase. After
    $1/\alpha$ such phases we have visited $3|S|/4$ vertices of $S$ in total. We can repeat this
    scheme for $O(\log |S|)$ phases to visit all vertices of $S$. Between any
    two consecutive phases we require $n$ time steps to reach the start vertex of the walk
    $W_{\sigma}$, as outlined in the previous paragraph. In total we get a
    temporal walk $W$ that spans $O((r|S|^{1+0.9\epsilon} + n) \log
    |S|)=O(r|S|^{1+\epsilon} + n \log |S|)$ time steps.
\end{proof}    
}
\full{\proofepsilonexploration}
By using the theorem above repeatedly for each orbit, we 
get a temporal walk for the whole temporal graph.

\begin{corollary}\label{cor:epsilonallorbits} Let $\mathcal{G}$ be a temporal
    graph with lifetime $\ell$ and vertex set $V$.
    For any fixed $\epsilon>0$, there exists a temporal walk
    $W$ that spans $O(r n^{1+\epsilon})$ time steps and visits all vertices of
    $V$, where $r=|\torb|$ is the orbit number.
\end{corollary}
\newcommand{\proofepsilonallorbits}{ 
\begin{proof}
    Visit the orbits one after the other and spend $n$ time steps to go from one
orbit to the next (via Lemma~\ref{lem:reachability}). Let the sizes of the orbits be $n_1,n_2,\ldots,n_r$.
By applying Theorem~\ref{thm:epsilonexploration},
all vertices of an orbit of size $n_i$ can be visited
in $O(rn_i^{1+\epsilon}+n\log n_i)
=O(rn_i^{1+\epsilon}+n\log n)$ time steps.
All $r$ orbits can therefore be explored in
$O(\sum_{i=1}^r (r n_i^{1+\epsilon} + n \log n))=
O(r n^{1+\epsilon}+rn\log n)=O(r n^{1+\epsilon})$ time steps,
where the first
transformation follows from $\sum n_i^p \le (\sum n_i)^p$ for any $p\ge
1$.
\end{proof}
}
\full{\proofepsilonallorbits}
\section{Upper Bound for TRP}
\label{sec:rendezvous}%
Using Theorem~\ref{thm:epsilonexploration} we show that \textsc{trp} can be
solved by constructing a walk that spans (asymptotically) the same number of
steps as a walk for exploring an arbitrary single orbit. The idea is that the
two agents identify an orbit in which they meet, and then the first agent moves
to this orbit, and after $n$ time steps the second agent starts exploring this
orbit. For this the agents must be able to independently identify the same
orbit, for which we introduce some additional notation. We extend the definition
of isomorphism to temporal graphs as follows. Let $\mathcal{G}, \mathcal{H}$ be
two temporal graphs with lifetime $\ell$ and vertex sets $V_{\mathcal{G}}$ and
$V_{\mathcal{H}}$, respectively. We call a bijection $\theta: V_{\mathcal{G}}
\rightarrow  V_{\mathcal{H}}$ a \textit{temporal isomorphism} if $\theta$ is an
isomorphism from $G_t$ to $H_t$ for each $t\in [\ell]$ (and thus also an
isomorphism from $G$ to $H$, which denote the underlying graphs of $\mathcal{G}$
and $\mathcal{H}$, respectively). If clear from the context, we say
\textit{isomorphism} instead of \textit{temporal isomorphism}. 

We define an \textit{integer coloring} as a coloring of the vertices in the
vertex set $V$ of a temporal graph $\mathcal{G}$ (with the colors being integer
values). The assigned colors induce a partial order $\prec_\mathcal{P}$ on the
vertex set $V$ such that, for all vertices $u, v \in V$ that are assigned colors
$c_u$ and $c_v$, respectively, with $c_u\neq c_v$ it holds that $u
\prec_{\mathcal{P}} v$ if $c_u < c_v$. The idea is now that the two agents
compute the same integer coloring of the given temporal graph $\mathcal{G}$ with
the property that two vertices $u, v \in V$ are assigned the same color $c$ if
and only if $u, v \in S$, with $S$ some orbit of $\torb$. The agents then meet
at the smallest orbit, breaking ties via the coloring.

Note that, since the agents do not have access to consistent labels of the
vertices in~$V$, they are unable to distinguish between two vertices $u, v \in
S$ with $S$ being an orbit. Intuitively, the two agents $a_1$ and $a_2$ view
$\mathcal{G}$ as different temporal graphs $\mathcal{G}_1$ and $\mathcal{G}_2$,
respectively, such that
$\mathcal{G}_1\cong \mathcal{G}\cong \mathcal{G}_2$. A natural
idea is for the agents to pick a smallest orbit for their meeting, but the
challenge is how to ensure that the agents pick the same orbit if there are
multiple equal-size orbits that all have the smallest size. Therefore, in the
proof of the following lemma we let the agents iterate over all possible
temporal graphs until they find a graph $\mathcal{H}$ with
$\mathcal{G}_1\cong \mathcal{H}\cong \mathcal{G}_2$. Then both
agents compute an integer coloring for $\mathcal{H}$ as outlined in the previous
paragraph. This coloring is translated to a coloring of $\mathcal{G}_1$ by
agent~$a_1$ and to a coloring of $\mathcal{G}_2$ by agent~$a_2$ via isomorphism
functions, which are independently computed by the agents.

\begin{lemma}\label{lem:meeting} Let $\mathcal{G}$ be a temporal graph with
    vertex set $V$ and lifetime $\ell$ and let $a_1, a_2$ be two label-oblivious
    agents. There exists a pair of programs $(p_1, p_2)$ assigned to $a_1$ and
    $a_2$, respectively, such that each agent computes the same integer coloring
    of $V$ and such that two vertices $u, v \in V$ have the same color exactly
    if $u, v$ are in the same orbit of $\mathcal{G}/{\text{Aut}(\mathcal{G})}$.
\end{lemma}
\newcommand{\proofmeeting}{ 
\begin{proof}
    Let $I$ be a list of all temporal graphs with $n$ vertices and lifetime
    $\ell$ such that no two entries of $I$ are isomorphic to each other. For
    every entry $\mathcal{H}_i$ of $I$ we denote with
    $\mathcal{A}_{\mathcal{H}_i}$ a sorted (in arbitrary but deterministic
    fashion) list of all orbits in $\mathcal{H}_i/\text{Aut}(\mathcal{H}_i)$.

    Assume that each of the program $p_1$ and $p_2$ computes these
    aforementioned structures. Note that both programs compute the exact same
    structures as the result is independent of the labels of the temporal graph
    $\mathcal{G}$ as viewed by agents $a_1, a_2$, respectively.

    Let $V_1$ and $V_2$ denote the vertex set of $\mathcal{G}$ with the labels
    seen by $p_1$ and $p_2$, respectively. As the agents are label-oblivious, a
    vertex in $V$ can have different labels in $V_1$ and in $V_2$. Denote by
    $\mathcal{G}_1$ and $\mathcal{G}_2$ the temporal graph $\mathcal{G}$ with
    these new labeled vertex sets $V_1$ and $V_2$, respectively. Keep in mind
    that $\mathcal{G}_1\cong \mathcal{G}\cong \mathcal{G}_2$.
    
    Now, both programs $p_1$ and $p_2$ find the first entry $\mathcal{H}$ in $I$
    for which $\mathcal{G}_1\cong \mathcal{H}$ and
    $\mathcal{G}_2\cong \mathcal{H}$, respectively. Denote with
    $\theta_1: V_1 \rightarrow V(\mathcal{H})$ and $\theta_2: V_2 \rightarrow
    V(\mathcal{H})$ arbitrary isomorphisms computed by $p_1$ and $p_2$,
    respectively. The intuitive goal is now that both agents compute an integer
    coloring for $\mathcal{H}$ as sketched previously, and this coloring is
    translated to a coloring of $\mathcal{G}$ via these isomorphism functions. 

    We now show how to assign the colors to vertices. For $k=1,2$, the program
    $p_k$ iterates over all $S_i \in \mathcal{H}/\text{Aut}(\mathcal{H})$ (which
    are precomputed and sorted in a deterministic but arbitrary fashion) and for
    each $S_i$ assigns the color $i$ to all vertices in $\{\theta_k^{-1}(u) \mid
    u \in S_i\}$. Note that both programs iterate over
    $\mathcal{H}/\text{Aut}(\mathcal{H})$ in the same order.

It remains to show that the coloring computed by $p_1$ is the same as the
coloring computed by $p_2$. Assume for a contradiction that a vertex $u$ in
$\mathcal{G}$ is colored differently by $p_1$ and $p_2$. Let $u_1$ and $u_2$ be
the names of $u$ in $\mathcal{G}_1$ and $\mathcal{G}_2$, respectively. Since $u$
is colored differently by $p_1$ and~$p_2$, vertices $v=\theta_1(u_1)$ and
$w=\theta_2(u_2)$ are in different orbits of
$\mathcal{H}/\text{Aut}(\mathcal{H})$. Moreover, since $u_1$ and $u_2$ is the
same vertex, there is an isomorphism $\sigma$ between $\mathcal{G}_1$ and
$\mathcal{G}_2$ that maps $u_1$ to $u_2$. Now we get an automorphism
$\theta_2\circ \sigma \circ \theta_1^{-1}$ of $\mathcal{H}$ that maps $v$ to
$w$, and thus $v$ and $w$ are in the same orbit of
$\mathcal{H}/\text{Aut}(\mathcal{H})$, a contradiction. 
\end{proof}
}
\full{\proofmeeting}
We can now easily construct an algorithm for \textsc{trp}.  The agents
simply meet in a smallest orbit, breaking ties via the integer coloring
(Lemma~\ref{lem:meeting}). The first agent moves to said orbit,
then the
second agent searches the orbit for the first agent. Note that the smallest
orbit has size at most $n/r$, where $r$ is the orbit number. Thus, the bound
on the number of time steps provided by Theorem~\ref{thm:epsilonexploration}
becomes $O(r(n/r)^{1+\epsilon}+n\log (n/r))= O(n^{1+\epsilon})$, and we obtain
the following upper bound for~$\trp$.

\begin{theorem}\label{thm:rendezvous-time} Let $\mathcal{G}$ be a temporal graph
    with lifetime $\ell$ and $a_1, a_2$ two label-oblivious agents. For any
    fixed $\epsilon > 0$, there exists a pair of programs $(p_1, p_2)$ assigned
    to $a_1, a_2$, respectively, such that the two agents are guaranteed to meet
    after $O(n^{1+\epsilon})$ time steps.
\end{theorem}

\section{Lower Bounds for TEXP and TRP}
\label{sec:lower-bounds}

We start this section with a simple lower bound for $\texp$, which is a
{fairly straightforward adaptation}
of the known lower bound of $\Omega(n^2)$ time steps~\cite{ErlebachHK21}.
Following that, we give a lower bound of $\Omega(n \log n)$ time steps for
$\trp$. For this we describe the construction of a temporal graph that is
connected and has only a single orbit. We then show how an adversary can choose
the starting positions of the two agents that want to meet in order to delay
their meeting. Intuitively, the graph we create is a cycle that changes
repeatedly after some number of steps. By our construction, the adversary can
make sure that after every change of the graph, the two agents are placed far
away from each other. In the end, we also show that the resulting lower bound
for $\trp$ yields a corresponding lower bound for $\texp$.

\begin{lemma}
\label{lem:rn-lower-bound}
    For any $1\le r\le n$, there exist $n$-vertex instances of $\texp$ with
    orbit number $r$ that require $\Omega(rn)$ time steps to be explored. 
\end{lemma}
\newcommand{\proornlowerbound}{ 
\begin{proof}
    Assume without loss of generality that $n$ is even. For $r=1$, the lower
    bound is trivial as any $n$-vertex temporal graph requires $\Omega(n)$ time
    steps to be explored. Next, assume $2\le r\le \frac{n}{2}$. The lower bound
    of $\Omega(n^2)$ presented by Erlebach et al.~\cite{ErlebachHK21} is based
    on the construction of a temporal graph $\mathcal{G}$ with vertex set $V$
    and lifetime $\ell$ such that at each time step $t \in [\ell]$ the graph
    $G_{t}$ forms a star as follows: the vertex set $V$ is partitioned into two
    sets $A$ and $B$, each of size~$\frac{n}{2}$. We modify this construction by
    letting $A$ consists of $r-1$ vertices and $B$ of the remaining vertices. In
    each of the first $r-1$ time steps, a different vertex of $A$ is chosen to
    be the center of the star. Vertices of $B$ are never centers of a star. This
    pattern repeats, i.e., in each time step $t>r-1$, the graph $G_t$ is equal
    to the graph $G_{t-(r-1)}$. The lifetime of the graph is set to $\ell=n^2$.
    By construction, all vertices of $A$ form their own orbit of size $1$ (each
    center is uniquely identified by the first time step where it is a center),
    and all vertices of $B$ form one orbit together. Thus, the graph has
    $(r-1)+1=r$ orbits. As shown by Erlebach et al., it takes $\Theta(r)$ time
    steps to go from one vertex of $B$ to another. Thus, visiting all vertices,
    in particular those in~$B$, takes $\Omega(r|B|)=\Omega(rn)$ time steps.
    Here, we have used that $r\le \frac{n}{2}$ implies $|B|\ge \frac{n}{2}+1$.

    {Finally, if $\frac{n}{2} <r \le n$, use the construction above for
    $r=\frac{n}{2}$ to create a temporal graph with $n^2$ time steps that
    requires $\Omega(n^2)$ time steps for exploration. Then, add one extra time
    step $t_f=n^2+1$ to split the orbit $B$ into the required number of orbits
    as follows. Assume we want to have orbit number $\frac{n}{2}+k$ for some
    $1\le k\le \frac{n}{2}$. Then $G_f$ consists of a path $P$ of
    $\frac{n}{2}+k-1$ vertices, made up of the $\frac{n}{2}-1$ vertices from $A$
    followed by $k$ vertices from~$B$, while the remaining $\frac{n}{2}+1-k$
    vertices of $B$ are attached as pendant vertices to the endpoint of the
    path~$P$ that lies in~$B$. The vertices of $B$ are split into $k+1$ orbits
    in this way: each of the vertices on $P$ forms an orbit by itself, while
    those attached as pendants to an endpoint of $P$ form one orbit together.
    Thus, the temporal graph has orbit number $\frac{n}{2}+k$, as desired.}
\end{proof}
}
\full{\proornlowerbound}

\begin{theorem}\label{thm:rendezvouslb}
    For any two agents $a_1$ and $a_2$ with arbitrary deterministic programs,
    there exist instances of $\trp$ where the agents
    require $\Omega(n \log n)$ time steps to meet.
\end{theorem}
\newcommand{\proofrendezvouslb}{ 
\begin{proof}
    We construct a temporal graph $\mathcal{G} = (G_1, \dots, G_\ell)$ with
    lifetime $\ell=n^2$ and vertex set $V = \{0, \dots, n-1\}$ with $n = 2^{m} -
    1$ for an arbitrarily large odd integer $m > 10$. We define the static graph
    $C_i$ with vertex set $V$ to be the cycle where each vertex $u \in V$ is
    adjacent to vertices $(u + n_i) \bmod n$ (the \emph{clockwise} neighbor) and
    $(u - n_i) \bmod n$ (the \emph{counter-clockwise} neighbor) for $n_i =
    2^{2i}$. {Note that $C_i$ is indeed a cycle because $n_i=2^{2i}$ and
    $n=2^m-1$ are coprime.} In what follows we use $n_i$ and $2^{2i}$
    interchangeably. For our construction (see Fig.~\ref{fig:lb} for an
    illustration), we will use the cycles $C_i$ for $i\in
    \{0\}\cup\{2,3,\ldots,\frac{m-7}{2}\}$, i.e., we will use the cycles $C_0,
    C_2,C_3, C_4,\ldots,C_{(m-7)/2}$. For each cycle $C_i$, we define the $k$-th
    \textit{section} of the cycle as $S_{i,k} = \{ u \mid u \in V \text{ and } u
    \bmod 2^{2i} = k \}$. There are $2^{2i}$ different sections in cycle $i$,
    one for each $k$ with $0\le k < 2^{2i}$. Each one contains $\lceil
    \frac{n}{2^{2i}} \rceil$ vertices, except the last one which contains
    $\lfloor \frac{n}{2^{2i}} \rfloor$ vertices due to parity reasons ($n$ is an
    odd number that is one smaller than a power of two). Intuitively, a section
    contains vertices with the same modulo value, and these vertices (in
    ascending order) form a consecutive part of the cycle~$C_i$. {Note that
    any $s$ sections contain together at least $s\cdot \frac{n}{2^{2i}} - 1$
    vertices.} The first and last vertex of a section is connected to the last
    vertex of the previous section and the first one of the next section,
    respectively.

    Furthermore, we define \textit{phase} $i$, for $i\ge 1$, of $\mathcal{G}$ as
    a number of consecutive time steps from $t\in [\ell]$ to $t' \in [\ell]$
    such that $G_{t''}=G_{t}$ for all $t \leq t'' \leq t'$ and $G_{t-1} \neq
    G_{t}$ (or $t=1$) and $G_{t'} \neq G_{t'+1}$. We set the \textit{duration}
    of each phase to $K = \lfloor n/16 \rfloor$. During phase $1$, the graphs
    $G_t$ are all equal to~$C_{0}$. During phase $i$ for $i>1$, the graphs $G_t$
    are all equal to $C_i$. We continue the construction for $(m-7)/2
    =\Theta(\log n)$ phases.
	
    Notice that the temporal graph $\mathcal{G}$ is constructed in such a way
    that it has only one orbit {as for every $d$ the function $\sigma$ with
    $\sigma(u)=(u+d)\bmod n$ for all $u\in V$ is an automorphism.} Hence, the
    vertices are indistinguishable to the two agents. Therefore, the
    deterministic programs that the agents run must execute a sequence of steps
    that is independent of the starting vertex. {(The adversary can let
    each agent see vertex labels that make the agent's view the same no matter
    in which vertex of $\mathcal{G}$ the agent is placed at time~$1$.)} The
    steps executed by an agent can be either edge traversals (clockwise or
    counter-clockwise) or waiting at a current vertex. In what follows, we show
    that independently of the sequence of steps that each agent's program
    performs, an adversary can choose the starting vertex for each agent in such
    a way that they will need $\Omega(n \log n)$ time steps to meet.
	
	\begin{figure}[hb!]
		\centering
		\includegraphics[width=\textwidth]{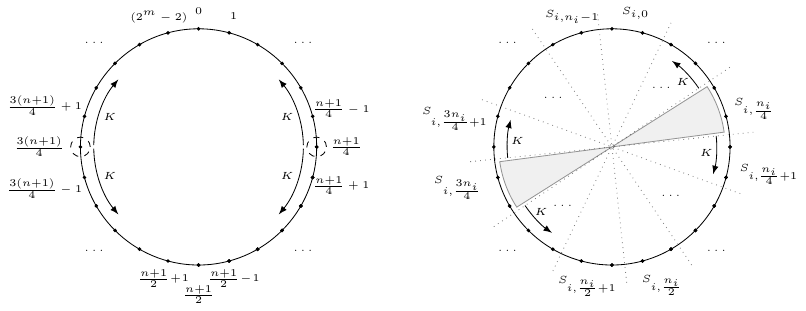}
		\caption{The initial cycle of phase $1$, $C_0$, on the left. Initially,
		the agents are placed on the vertices marked with dashed circles and can
		move up to $K$ vertices away in phase~$1$. In each other phase $i \geq
		2$, the cycle $C_i$ is viewed as consisting of sections, as shown on the
		right. At the beginning of the phase, the agents are placed in vertices
		of two sections (shown grey) that are sufficiently far apart, and they
		can move up to $K$ vertices away from the boundary of their section
		before the next phase starts.}
		\label{fig:lb}
	\end{figure}
	
    We are going to use bit vectors of length $m = \lceil \log n \rceil$ to
    represent a vertex, i.e., the bit vector $\vec{b}=(b_{m -1} b_{m - 2} \dots
    b_1 b_0)$ refers to the vertex whose binary representation is~$\vec{b}$.
    Initially, we place the agent $a_1$ at vertex $u = \frac{1}{4}(n + 1)$ and
    agent $a_2$ at vertex $v=\frac{3}{4}(n + 1)$ of $C_0$. Note that the binary
    representation of $u$ is $(010\cdots0)$ and that of $v$ is $(110\cdots0)$,
    i.e., the lowest $m-2$ bits of both bit vectors are all zero.
    
    During phase~$1$, each agent can move up to $K$ vertices away from its start
    vertex. Given the initial placement of the agents and the size of the cycle,
    it is guaranteed that they cannot meet in phase~$1$: {Their initial
    distance is roughly $n/2$, and even if both of them move $K$ steps closer to
    the other agent, their distance will only reduce by $2K\le n/8$.
    Furthermore, even if we shift each agent's start vertex by at most $n/8$
    positions, they still cannot meet in phase~$1$: Their initial distance will
    be at least roughly $n/4$, and their distance can reduce by at most $n/8$ in
    phase~$1$.  We will actually shift the start positions of both agents only by
    fixing the lowest $m-5$ bits of their bit vectors, so their start positions
    will be shifted by at most $2^{m-5}-1\le n/32$, and it is clear that the
    agents will not be able to meet in phase~$1$ no matter what the exact amount
    of shift is. Our strategy will be to fix the lower-order bits of the start
    positions of the agents in such a way that the agents cannot meet in
    phase~$i$ for $2\le i\le \frac{m-7}{2}$ either.}

    By adjusting the lowest $4$ bits of $u$ and $v$, i.e., bits $b_0, b_1, b_2,
    b_3$ of their respective bit vectors, we can ensure that agent $a_1$ ends
    phase~$1$ at a vertex $u_1$ with $u_{1} \bmod {16} = 4$ and agent $a_2$ ends
    phase~$1$ at a vertex $u_2$ with $u_{2} \bmod {16} = 12$. This is because
    fixing the bits $b_0, b_1, b_2, b_3$ in one of the 16 possible ways will add
    an offset $\Delta\in\{0,1,2,\ldots,15\}$ to the start position of an agent,
    and the same offset is added to the position in which the agent ends
    phase~$1$ (because every edge traversal adds or subtracts a fixed value to
    the current position). As the agents end phase~$1$ in vertices $u_1$ and
    $u_2$ satisfying $u_{1} \bmod {16} = 4$ and $u_{2} \bmod {16} = 12$, the
    first agent is placed in a vertex in section $S_{2, \frac{1}{4} 2^{2 \cdot
    2}} = S_{2, 4}$ of $C_2$ while the second agent is placed in a vertex in
    section $S_{2,  \frac{3}{4} 2^{2 \cdot 2}} = S_{2, 12}$ of $C_2$ in the
    beginning of phase~$2$. Sections $S_{2,4}$ and $S_{2, 12}$ are separated by
    seven other sections, which together span at least $7\cdot \frac{n}{16} -
    1\ge n/4$ vertices. Therefore, the agents cannot meet during phase~$2$.

    Now, depending on the sequence of steps taken by the two agents in
    phase~$2$, we will show that it suffices for the adversary to fix the next
    two bits, i.e., bit $b_4$ and $b_5$ of the bit vectors of the start
    positions of the two agents, in order to guarantee that the two agents are
    placed sufficiently far from each other at the beginning of phase~$3$ and
    therefore cannot meet in that phase either. Observe that fixing bits $b_4$
    and $b_5$ does not change the value of $u_1\bmod 16$ and $u_2\bmod 16$.
    Thus, the sections in which the agents begin phase~$2$ are not altered by
    fixing bits $b_4$ and~$b_5$. This process then repeats for all phases $i>2$:
    Depending on the sequence of steps taken by the agents in phase~$i$, the
    adversary fixes bits $b_{2i}$ and $b_{2i+1}$ of the binary representations
    in such a way that the agents begin phase~$i+1$ in sections of $C_{i+1}$
    that are sufficiently far from each other (without altering the sections in
    which the agents have started any of the previous phases).

    We have already shown that the two agents cannot meet in phase~$1$ or~$2$.
    In the following, we show how to fix bits $b_{2i}$ and $b_{2i+1}$ at the end
    of phase~$i$ for $i\ge 2$ in such a way that the two agents cannot meet in
    phase $i+1$ (in cycle $C_{i+1}$) either.
	
   At the end of phase~$i$, we have four options to choose from for the section
   of $C_{i+1}$ where we would like agent $a_1$ to start phase~$i+1$. Let $z$ be
   the value determined by the lowest $2i$ bits of $a_1$'s current position (at
   the end of phase~$i$). Note that $z\le 2^{2i}-1<\frac14\cdot
   2^{2(i+1)}=\frac14 n_{i+1}$. The four choices $00$, $01$, $10$ and $11$ for
   bits $b_{2i}$ and $b_{2i+1}$ of $a_1$'s start position produce positions for
   $a_1$ at the end of phase~$i$ that are equal to $z$, $z+2^{2i}$, $z+2\cdot
   2^{2i}$, and $z+3\cdot 2^{2i}$ modulo $2^{2(i+1)}$, corresponding to sections
   $S_{i+1,z}$, $S_{i+1,z+\frac14 n_{i+1}}$, $S_{i+1,z+\frac12 n_{i+1}}$,
   $S_{i+1,z+\frac34 n_{i+1}}$ of $C_{i+1}$. Analogously, there is a $z'<\frac14
   n_i$ such that agent~$a_2$ can be placed in any of the four sections
   $S_{i+1,z'}$, $S_{i+1,z'+\frac14 n_{i+1}}$, $S_{i+1,z'+\frac12 n_{i+1}}$,
   $S_{i+1,z'+\frac34 n_{i+1}}$ of $C_{i+1}$ at the start of phase~$i+1$. We
   choose bits $b_{2i}$ and $b_{2i+1}$ for $a_1$ and $a_2$ in such a way that
   $a_1$ begins phase $i+1$ in $S_{i+1,z}$ and $a_2$ begins phase $i+1$ in
   $S_{i+1,z'+\frac12 n_{i+1}}$. The sections in which the two agents start
   phase $i+1$ are then separated by at least $\frac14 n_{i+1}$ sections in
   $C_{i+1}$, and these $\frac14 n_{i+1}$ sections contain together at least
   $\frac14 n_{i+1} \cdot \frac{n}{n_{i+1}} - 1 = \frac{n}{4}-1$ vertices. As
   each agent can move over at most $K=\lfloor n/16\rfloor$ vertices in phase
   $i+1$, it is impossible for the agents to meet in phase~$i+1$.

    We iterate this construction until phase $\frac{m-7}{2}$. (For all time
    steps from the end of phase $\frac{m-7}{2}$ until the lifetime $\ell=n^2$ is
    reached, we can let the graph be equal to $C_0$.) By construction, the
    agents cannot meet in any of the phases from $i$ to $\frac{m-7}{2}$. As we
    have fixed the four lowest bits of the start positions of the two agents at
    the end of phase~$1$ and fixed two further bits in each of the remaining
    $\frac{m-9}{2}$ phases, we have fixed the lowest $4+(m-9)=m-5$ bits of the
    start positions in this way. As observed at the start of the proof, this
    ensures that the agents cannot meet in phase~$1$ either. Hence, there are
    $\frac{m-7}{2}=\Theta(\log n)$ phases, each with $\lfloor n/16\rfloor$ time
    steps, during which the adversary can prevent the agents from meeting. This
    shows that $\Omega(n\log n)$ time steps are needed for the agents to meet.
\end{proof}
}
\full{\proofrendezvouslb}
\begin{corollary}\label{cor:explorationlb} There exist connected temporal graphs
    $\mathcal{G}$ with vertex set $V$, lifetime $\ell$ and a single orbit such
    that all temporal walks $W$ require $\Omega(n \log n)$ time steps to visit
    all vertices of $V$.
\end{corollary}
\newcommand{\proofexplorationlb}{ 
\begin{proof}
    Assume for a contradiction that for every connected temporal graph
    $\mathcal{G}$ with $n$ vertices and a single orbit there exists a temporal
    walk that visits all vertices in $o(n\log n)$ time steps. This implies that
    \textsc{trp} can be solved in $o(n\log n)$ time steps for connected
    single-orbit temporal graphs: Let the agent $a_1$ perform a temporal walk
    that visits all vertices in $o(n\log n)$ time steps, and let the agent $a_2$
    remain at its initial vertex. By Theorem~\ref{thm:rendezvouslb}, no such
    solution to \textsc{trp} is possible.
\end{proof}
}
\full{\proofexplorationlb}
{The lower bounds of Theorem~\ref{thm:rendezvouslb} and
Corollary~\ref{cor:explorationlb} can be adapted to temporal graphs with $r$
orbits for any constant~$r$ as follows: Use the construction from the proof of
Theorem~\ref{thm:rendezvouslb}, but instead of letting the graph $G_t$ in each
time step be a single cycle~$C_i$, let $G_t$ contain $2r-1$ copies of~$C_i$, and
for each vertex $u$ of $C_i$ connect all copies of $u$ by a path $P_u$ (starting
with the vertex $u$ in the first copy of $C_i$ and ending with the vertex $u$ in
the $(2r-1)$-th copy of $C_i$). The resulting temporal graph has $r$ orbits: The
vertices of the `middle' copy of~$C_i$ form one orbit, and the vertices in the
two copies of $C_i$ that have distance $k$ from the middle copy, for $1\le k\le
r-1$, also form an orbit. Let $n'=n/r$ denote the number of vertices in one copy
of~$C_i$. By the arguments in the proof of Theorem~\ref{thm:rendezvouslb}, it
takes $\Omega(n'\log n')$ time steps for the two agents to reach a location in
the same path $P_u$, and thus $\trp$ requires $\Omega(n'\log n')$ time steps. As
$n'=n/r$ and $r$ is a constant, this gives a lower bound of $\Omega(n\log n)$
for $\trp$. The lower bound of $\Omega(n \log n)$ time steps for exploration of
temporal graphs with orbit number $r$ for any constant $r$ then follows as in
the proof of Corollary~\ref{cor:explorationlb}.}

\full{
\section{Efficient Computations via Randomization for Exploring
Orbits}\label{sec:rand}

In this section, we present a randomized algorithm for exploring a given orbit
$S \in \torb$ of a temporal graph $\mathcal{G}$ with high probability. Our goal
is to construct {efficiently} the temporal walk of Theorem~\ref{thm:orbitexplore} that visits
all vertices of $S$. {Recall that this walk spans $O((n^{5/3}+rn)\log n)$ time steps.} Our previously outlined algorithms for exploring the
orbit $S$ all use a subroutine that first constructs some temporal walk $W$ that
visits a certain fraction of the vertices of $S$, followed by iterating over all
automorphisms $\sigma \in \auttg$ and applying each $\sigma$ to $W$ until a
temporal walk $W'$ is obtained that visits sufficiently many previously unvisited vertices.
{The runtime of this routine is $\Omega(|\auttg| \cdot |W|)$ in addition to
computing the set of all automorphisms, unless they are given as an input. Our
goal in this section is to show an algorithm with a runtime that is independent
of $|\auttg|$ that computes our desired temporal walk. {
Our algorithm has a runtime that consists of $O(n^{1/3}\log(n/\epsilon))$
linear-time {considerations} of the snapshots that exist during the timesteps
of the computed temporal walk. As an additional result we give an algorithm
that finds the aformentioned temporal walk with probability $1$, with
the caveat of an expected runtime. {More exactly, the latter} algorithm has
an {expected runtime} that
consists of $O(n^{1/3})$ expected linear-time {considerations}.}

Unlike the deterministic algorithms of the previous sections, the randomized
algorithm does not require the set of all automorphisms as input since the
single orbit $S \in \autg$ we want to explore is sufficient. Recall that in
Lemma~\ref{lem:transform1}, for a given set {$T \subseteq S$} and temporal walk $W$
that visits $h$ vertices of $S$, there is a certain fraction of the vertices of
$S \setminus T$ that we are guaranteed to visit when applying some automorphism
$\sigma \in \auttg$ to $W$. This guaranteed fraction is the same as the expected
value of the fraction obtained by applying a random automorphism to $W$.}
Intuitively, we virtually construct a special version of the automorphism matrix
from Section~\ref{sec:aut-util} and uniformly at random choose a row~from~it.

To describe our algorithms we introduce some additional notation. Let
$\mathcal{G}$ be a connected temporal graph with lifetime $\ell$ and let $S \in
\torb$ be any orbit.  We define a
\emph{$(S, t, k)$-greedy walk $W$} as a temporal walk that starts at a given
time step {$t \in [\ell]$} at a vertex $u \in S$ and visits {$k \in [|S|]$} different vertices of
{$S$ such that each of the $k$ vertices is visited at the earliest
arrival time. More exactly, the second vertex $u'$} of the $k$ vertices is visited
by a temporal walk $W'$ starting at vertex $u$ in time step $t$ such that no
temporal walk exists that starts at vertex $u$ in time step $t$ and that visits
any vertex of $S \setminus \{u\}$ at an earlier time step (i.e., $W'$ is a {\em
foremost temporal walk}), and the next vertex $u'' \in S \setminus \{u, u'\}$ is
visited by a foremost temporal walk $W''$ starting at vertex $u'$ in time step
$t'$, and so on until the $k$ different vertices are visited. Intuitively, it is
a greedy strategy to visit $k$ different vertices of $S$ in a fast way. By
repeatedly applying Lemma~\ref{lem:nextreachable} until $k$ different vertices
are visited, we can upper bound the time steps required by any $(S, t,
k)$-greedy walk by $O(k(k n/|S| + r))$, where $r=|\torb|$ is the orbit number. 
{Note that the bound of Lemma~\ref{lem:nextreachable} is stated
as an upper bound for the earliest time step that $k$ different vertices are
reachable. The same idea}
was previously used in the proof of Theorem~\ref{thm:orbitexplore}. 
We refer to the upper bound of required time steps by any $(S, t, k)$-greedy walk
as the \emph{$(S, k)$-greedy number $\Delta_{(S, k)}$}. 
{Let $T \subseteq S$ with $p=(|S|-|T|)/|S|$.}
We define an
\emph{$(S, T, t, k)$-random walk $W'$} as a temporal walk that starts in time
step $t$ at a vertex $u \in S$, is expected to visit $\geq p k$ vertices of $S
\setminus T$ for any $k \leq |S|$, and is sampled from a
uniform distribution of all $(S, t, k)$-greedy walks. Note that the distribution has
an expected value of visiting $p k$ vertices of $S \setminus T$. Finally, we
define $R_{t, t'}(u)$ as the set of vertices reachable with a temporal walk
starting at a {vertex $u$} in time step $t$ and ending in time step $t'$,
for any given {$u \in V$  and} $t, t' \in [\ell]$.

We can compute foremost temporal walks using the results of Wu et
al.~\cite{WuCYSYW16}. They give various results (Algorithm~1-3) that can be
further simplified for our use case. In essence, their algorithm looks at all
temporal edges that exist in a given interval of time steps $[t, t']$ and
iteratively computes the set $R_{t, t'}(u)$ while maintaining a tree-like
structure that stores all computed temporal walks, which is all that we require
throughout this section. Their algorithm runs in linear time $O(\sum_{t'' \in
[t, t']} |E(G_{t''})|)$. Since $(S, t, k)$-greedy walks are simply an iterative
application of {their} algorithm, they can also be computed in linear time.
Algorithm~\ref{alg:one} shows a version of this with added randomization.

\begin{lemma}\label{lem:rand}
    Let $\mathcal{G}$ be a temporal graph with vertex set $V$, lifetime $\ell$
    and  \scalebox{0.95}{$S \in \torb$} any orbit. Via Algorithm~\ref{alg:one} we can construct an $(S, T,
    t, k)$-random walk in time $O(\sum_{t'' \in [t, t+\Delta_{(S, k)}]} |E(G_{t'})|)$ with
    $\Delta_{(S, k)} =O(k(k n/|S| + r))$ the $(S, k)$-greedy number and $r=|\torb|$.
\end{lemma} 

\begin{algorithm}[h!]
    \caption{Algorithm that constructs an $(S, T, t, k)$-random walk $W$ in time
    linear in the number of temporal edges that exist between time steps $t$
    and $t+\Delta_{(S, k)}$ of the input temporal graph $\mathcal{G}$,
    with $\Delta_{(S, k)}=O(k(k n/|S| + r))$ the $(S, k)$-greedy number
    and $r=|\torb|$.}
    \label{alg:one}
    \DontPrintSemicolon
    \SetKwInOut{Input}{In}
    \SetKwInOut{Output}{Out}

    \Input{Temporal graph $\mathcal{G}$ with lifetime $\ell$, time step $t \in
    [\ell]$, orbit $S \in \torb$, some $k \leq |S|$ and $T \subseteq S$. }

    \Output{$(S, T, t, k)$-random walk $W$.}

    Choose $u$ uniformly at random from $S$\\
    $X:=\{u\}$ \\
    $W:=\emptyset$\\
    $t' := t$\\
    \While{$|X| < k$}
    {
        $t''=$ minimum time step such that $(R_{t', t''}(u) \cap S) \setminus X \neq \emptyset$.\\
        Choose $v$ uniformly at random from $(R_{t', t''}(u) \cap S) \setminus X$\\
        $W':=$ temporal walk that starts at $u$ in time step $t'$
        and ends at vertex $v$ in time step $t''$ (earliest arrival time)\\
        $W := $ $W$ extended by $W'$\\
        $t':= t''$\\
        $X:=X \cup \{v\}$\\
        $u:=v$
    }
    \Return $W$
\end{algorithm}

\begin{proof} 
    Let $\mathcal{W}$ be the set of all $(S, t, k)$-greedy walks. We define a
    matrix $M$ such that each row is an entry $W \in \mathcal{W}$. Note that for
    any two vertices $u, v \in S$, the set of all $(S, t, k)$-greedy walks
    starting at $u$ is isomorphic to the set of all $(S, t, k)$-greedy walks
    starting at $v$.
    {For simplicity, we remove all columns of $M$ containing vertices of
     $V \setminus S$ as well as vertices of $S$ that occur a second, third, etc.\ time. 
     The reason for this is that we make no statements about these vertices in the
following.
     {Further note} that, by our automorphism funcion, there are no colums that
contain vertices of both, $V \setminus S$ and $S$. For the same reason, if a vertex
of some column $k'$ occurs a second time in a row, then all vertices of
column $k'$ occur for a second time in their row.}
%
%
%
    {Thus,}
    $M$ now {consists of} $k$ columns containing only vertices of $S$, with each row
    describing a sequence of $k$ vertices of $S$ visited {in the
    same order} by some $(S, t, k)$-greedy
    walk. Following the proof of Lemma~\ref{lem:transform1}, we know that each
    column contains $c$ copies of each vertex of $S$ for some $c > 0$ (i.e., all
    vertices have the same non-zero number of occurrences in each column). It
    follows that uniformly at random choosing a row of $M$ results in
    {a row} $\rho$ containing $p k$ vertices of $S \setminus T$ on average. We denote by
    $M_{u}$ for $u \in S$ the submatrix of $M$ restricted to rows containing the vertex $u$
    as the first entry. We extend this to submatrices $M_{(u_1, u_2, \ldots)}$
    restricted by arbitrary sequences $u_1, u_2, \ldots$ containing no duplicate
    vertices such that $M_{(u_1, u_2, \ldots)}$ contains $u_1$ in the first column
    and $u_2$ in the second column, and so forth.
    Note that all possible such non-empty submatrices for sequences of the same
    length are isomorphic to each other. We now show that Algorithm~\ref{alg:one}
    effectively chooses a row of $M$ uniformly at random. The algorithm begins by
    choosing uniformly at random a vertex $u \in S$ as the start vertex. This is
    equivalent to constructing the matrix $M_u$. Since $u$ is chosen uniformly at
    random, all submatrices $M_{v}$ for $v \in S$ have the same probability of being chosen
    in this fashion.
    
    The next step of Algorithm~\ref{alg:one} is to construct a
    temporal walk $W'$ that visits a vertex of $S \setminus \{u\}$, in particular
    with the earliest arrival time. Let $S'$ be the non-empty set of vertices
    reachable at the earliest time steps by a temporal walk starting from $u$ at
    time step $t$. The matrix $M_u$ contains exactly all vertices of $S'$ in its
    second column, and each vertex of $S'$ is contained the same number of times in
    this column. Uniformly at random choosing a vertex $u' \in S'$ is then
    equivalent to choosing a submatrix $M_{(u, u')}$ such that every row in $M_u$
    has the same probability of being part of $M_{(u, u')}$.
    Each step of the algorithm repeats this to construct submatrices with longer
    sequences until a sequence of $k$ different vertices is chosen. At this
    point the submatrix contains only one row, chosen uniformly at random from
    the entire matrix $M$. This gives us the desired $(S, T, t, k)$-random walk
    $W$. 
    
    Finally, we analyze the runtime. Because the resulting temporal walk
    $W$ is an $(S, t, k)$-greedy walk, the number of time steps it spans is
    bounded by the $(S, k)$-greedy number $\Delta_{(S, k)} =O(k(k n/|S| + r))$ with
    $r=|\torb|$. The algorithm can easily be implemented by iteratively
    constructing the sets of reachable vertices for each considered time step,
    i.e., one implements this via a simple algorithm that iterates over all
    temporal edges at each considered time step~\cite{WuCYSYW16}, the runtime of
    the algorithm is $O(\sum_{{t'} \in [t, t+\Delta_{(S, k)}]} |E(G_{t'})|)$.
\end{proof}

    By repeatedly running the randomized algorithm of Lemma~\ref{lem:rand}, one can
achieve the following lemma. 
{
\begin{lemma}\label{lem:randlog} Let $\mathcal{G}$ be a connected temporal
    graph with lifetime $\ell$, let $S \in \torb$ be any orbit, and let $T
    \subseteq S$ with $p=(|S|-|T|)/|S|$. For $k \leq |S|/2$,  $0 < \epsilon < 1$ and
    $r=|\torb|$ the orbit number, we
    can construct an $(S, T, t, 2k)$-random walk $W$ that visits $\geq p k$
    vertices of $S \setminus T$ with probability at least $1-\epsilon$. 
    The walk $W$ ends in time step $t'=t+\Delta_{(S, 2k)}$, with
    $\Delta_{(S, 2k)}=O(k(k n/|S|+r))$ being the $(S, 2k)$-greedy number.
    We construct $W$ by 
    sampling 
    {$\lceil \log\epsilon/\log(1-p/2) \rceil$-many
    $(S, T, t, 2k)$-random walks
    and returning one
    that visits $\geq pk$ vertices of $S\setminus T$, or indicating
    that no such walk was found among the samples}.
    The runtime of this procedure is $O((\sum_{t'' \in [t, t']}|E(G_{t''})|)
\log \epsilon/\log(1-p/2))$. 
\end{lemma}
\begin{proof}
    Let $k' \leq |S|$ be an integer and $0 < \theta < 1$.
    We begin by showing that the
    probability of an $(S, T, t, k')$-random walk $W'$ visiting at least $\theta
    pk'$ vertices of $S\setminus T$ is at least $(1-\theta)p$. 
    Let $\mathcal{W}$ be the set of all $(S, t, k')$-greedy walks
    from which we
    sample a walk to obtain an $(S, T, t, k')$-random walk. We call a walk $W'
    \in \mathcal{W}$ a \emph{bad walk} if $W'$ visits $< \theta p k'$ vertices
    of $S\setminus T$, and a \emph{good walk} if $W'$ visits $\geq \theta p k'$
    vertices. 

    Assume that {the 
    vertices} of $T$ are red, and the vertices of $S\setminus T$ are
    green. In each $W' \in \mathcal{W}$ we have $k'$ colored vertices. 
    %
    Let $w=|\mathcal{W}|$. In total over all {walks} in $\mathcal{W}$, we have $pk'w$ green and $(1-p)k'w$
    red vertices. A bad walk must have less than $\theta pk'$ green vertices,
    and thus at least $k'-\theta pk'=(1-\theta p)k'$ 
    red vertices. Thus, the
    number of bad walks is upper bounded by {$\frac{(1-p)k' w}{(1-\theta p)k'}$},
    which can be simplified as follows:
    \begin{align*}
        w\frac{(1-p)k'}{(1-\theta p)k'}  &= w\frac{1-p}{1-\theta p} \\
                                        &= w\frac{1-\theta p - (1-\theta)p}{1-\theta p} \\
                                        &= w(1-\frac{(1-\theta)p}{1-\theta p}) \\
                                        &\leq w(1-(1-\theta)p)
    \end{align*}
    Thus, we have at most $w(1-(1-\theta)p)$ bad walks, and at least
    $w(1-\theta)p$ good walks. It follows that the probability of an $(S, T,
    t, k')$-random walk visiting $\geq \theta pk'$ vertices of $S\setminus T$
    is $\geq (1-\theta) p$. Choosing $\theta = 1/2$ and $k'=2k$ we have that an
    $(S, T, t, 2k)$-random walk visits $\geq pk$ vertices of $S \setminus T$
    with probability at least $p/2$.

    We now analyze the number $N$ of $(S, T, t, 2k)$-random walks that we have to construct to get a
    probability of $1-\epsilon$ of {at least one sampled walk visiting $\geq pk$} vertices of $S\setminus
    T$. Let $W_1, \ldots, W_N$ be a set of sampled {$(S, T, t, 2k)$}-random walks.
    The probability that at least one of the $N$ temporal walks $W_i$ for $i \in
    [N]$ visits at least $pk$ different vertices is at least $1 - (1-p/2)^N$.
    So we have to solve $1 - (1-p/2)^N \geq 1 - \epsilon$, {
    or equivalently
    $(1-p/2)^N \leq \epsilon$ or $N\geq \log(\epsilon)/\log(1-p/2)$.}
%

    Thus, {$N = \lceil \log \epsilon/\log(1-p/2) \rceil$} many
    samples are sufficient to achieve our desired probability and our stated
    runtime is a result of repeating Algorithm~\ref{alg:one} at most $N$ times.
\end{proof}
}
{We now present a simple corollary that gives us the temporal walk
    of Lemma~\ref{lem:randlog} with probability $1$, but only expected runtime.
    First we set $\epsilon=1/2$. In this way we get a walk
    that visits $pk$ vertices of $S\setminus T$ with probability $1/2$. 
    We expect to run the algorithm of the lemma twice so that
    we get an
    expected runtime of $O((\sum_{t'' \in [t, t']} |E(G_t)|)/ \log(1/(1-p/2)))$.

\begin{corollary}\label{cor:expfraction}
    Let $\mathcal{G}$ be a connected temporal graph with lifetime $\ell$, let
    $S \in \torb$ be any orbit, and let $T \subseteq S$ with $p=(|S|-|T|)/|S|$.
    For $k \leq |S|/2$ and $r=|\torb|$ the orbit number, we can construct a
    temporal walk $W$ that visits $\geq p k$ vertices of $S \setminus T$. The
    walk $W$ ends in time step $t'=t+\Delta_{(S, 2k)}$ where $\Delta_{(S,
    2k)}=O(k(k n/|S|+r))$ is the $(S, 2k)$-greedy number. 
    {In expectation, we construct $W$ by 
    sampling $O(1/\log(1/(1-p/2)))$ many $(S, T, t, 2k)$-random walks, i.e.,
    we have}
    an expected running time of $O((\sum_{t'' \in [t, t']}|E(G_{t''})|)
    /\log(1/(1-p/2))$.
\end{corollary}
}

Finally, we present our randomized algorithm for computing the temporal walk of
Theorem~\ref{thm:orbitexplore}, i.e., exploring all vertices of a given orbit
$S$. Our randomized algorithm computes the temporal walk with high probability
in deterministic time. Recall that the construction scheme used
for the temporal walk of Theorem~\ref{thm:orbitexplore} consists of
$O(n^{2/3}\log n)$ iterative phases. In each phase, there are two subroutines
that are executed. First, the set of reachable vertices is expanded, starting at
the current time step $t$ and at the current vertex $u \in S$, until we are able
to construct a temporal walk to {every vertex} of the orbit $S$. Let $t'$ be the
earliest such time step. To concretely compute this set of reachable vertices
(and time step $t'$), we can use the algorithm of~\cite{WuCYSYW16}, which runs
linearly in the number of temporal edges that exist in the time span $[t, t']$.
The second phase then computes a temporal walk $W'$ from an arbitrary vertex $u
\in S$ starting at time step $t'$ that visits a certain number of vertices of
$S$, not necessarily previously unvisited. {The algorithm of Theorem~\ref{thm:orbitexplore} then
iterates} over all
automorphisms and {applies} them to $W'$ until the resulting temporal walk $W''$
visits enough previously unvisited vertices. We then compute a temporal walk
$W_u$ starting at vertex $u$ in time step $t$ and ending at the first vertex of
$W''$ in time step $t'$. We then concatenate $W''$ with the temporal walk $W_u$.

Now, instead of iterating over all automorphisms to compute the path $W''$, we
use Lemma~\ref{lem:randlog} to compute 
{a suitable}
temporal walk $W''$ with
high probability. {To achieve a high probability once only few
    vertices remain unvisited, we require a large amount of samples. E.g.,
    once there are only a constant number of unvisited vertices, we would
    require $\Omega(n)$ samples. To mitigate this, we make use of the fact
    that, if there are $O(n^{2/3})$ vertices left to be visited, we can visit them with
a trivial deterministic algorithm that computes a walk that requires $O(n^{5/3})$
time steps.}
The following theorem summarizes the result, and a detailed
analysis regarding the probability of success is contained in the proof.

\begin{theorem}\label{thm:randorbitexplore} Let $\mathcal{G}$ be a temporal
    graph with lifetime $\ell$ and vertex set $V$. Take $S \in \torb$ and
    $r=|\torb|$ the orbit number. For any $t\in[\ell]$, there exists a
    temporal walk $W$ starting at time step $t$ that visits all vertices of
    $S$ with probability $1-\epsilon$ and ends at time step
    $t'=t+O((n^{5/3}+rn)\log n)$. The temporal walk $W$ can be computed in
    time $O((\sum_{t'' \in [t, t']} |E(G_{t''})|) n^{1/3}\log (n/\epsilon))$, {i.e.,
    in time $O( M_{t,t'} \ n^{1/3} \log (n/\epsilon))$ where $M_{t,t'}$ is the total number
    of time edges of snapshots $G_t, \ldots, G_{t'}$.}
\end{theorem} 
\begin{proof} 
    Note that we assume that $S$ contains $\Omega(n^{2/3})$ vertices since,
    otherwise, we can visit all vertices of $S$ with a trivial deterministic algorithm that
    computes a walk that spans $O(n^{5/3})$ time steps. As mentioned, the
    algorithm of Theorem~\ref{thm:orbitexplore} consists of $O(n^{2/3}\log
    n)$ phases. In each phase, a temporal {walk $W'$} is constructed that visits
    $k\geq \lceil n^{1/3} \rceil$ vertices of $S$ such that {$W'$ visits} at least
    $pk$ vertices of $S \setminus T$ where $T$ is the set of previously visited
    vertices of $S$ and $p=(|S|-|T|)/|S|$. The number of phases is chosen so
    that only a constant number of vertices remain in $S\setminus T$. We modify
    the algorithm slightly by using at most $\pi(n)=c n^{2/3} \log n$ phases
    for some constant $c$ such that $\Theta(n^{2/3})$ unvisited vertices
    remain, which
    we can then visit by a trivial deterministic algorithm. Thus, we can  bound the value 
    {$p=(|S|-|T|)/|S| \geq 2/n^{1/3}$} throughout all phases.

    To find a partial temporal walk that visits $pk$ vertices {of $S\setminus T$} in each phase, we
    run the algorithm of Lemma~\ref{lem:randlog}. Recall that the algorithm
    samples {$\lceil \log \epsilon'/\log(1-p/2) \rceil$ many $(S, T, t, 2k)$-random
    walks} such that the
    probability that at least one such sample visits $pk$ vertices of $S\setminus T$ is $1 -
    \epsilon'$. Since we have $\pi(n)$ phases, {the probability of succeeding 
    in all}
    phases is $f(\pi(n))=(1-\epsilon')^{\pi(n)}$. 
    To get a probability of {$(1-\epsilon)$ as stated in this theorem}, we 
    choose
    $\epsilon'=\epsilon/\pi(n)$ since then we can show by a simple mathematical
    analysis that 
    $f(\pi(n))=(1-\epsilon/\pi(n))^{\pi(n)} \geq 1-\epsilon$ for $\pi(n)
    = 1$ and, since $f$ is a monotonely increasing function, for all $\pi(n) >
    1$.

    We now analyze the number of samples drawn by each execution of
    Lemma~\ref{lem:randlog} with our chosen value for $\epsilon'$ and our bound
    for $p$. We {draw $\lceil \log (\epsilon')/\log(1-p/2) \rceil$} 
    samples.
    {In the following we assume without loss of generality that $n>1$}. By the Mercator series, $\log (1-x) = -x -x^2/2 -x^3/3 - \ldots< -x$ for all 
    $0<x<1$, we can simplify our bound on the number of samples as follows. 
    \begin{align*}
         \lceil \log (\epsilon')/\log({1-p/2}) \rceil &\leq \lceil \log (\epsilon/\pi(n))/\log(1-1/n^{1/3})\rceil\\
                                              &= - \lceil \log (\pi(n)/\epsilon)/\log(1-1/n^{1/3})\rceil\\
                                                 &< - \lceil \log (\pi(n)/\epsilon)/(-1/n^{1/3})\rceil\\
                                                 &=  \lceil \log (\pi(n)/\epsilon)n^{1/3} \rceil\\
                                                 &=  \lceil \log (c \cdot n^{2/3} \log n/\epsilon) n^{1/3} \rceil\\
                                                 &{< \lceil \log (c \cdot n/\epsilon) n^{1/3} \rceil}
    \end{align*}
    We now show how we arrive at our final runtime. In each phase we use
    Lemma~\ref{lem:randlog} to draw {$\lceil \log (c \cdot n/\epsilon) n^{1/3} \rceil$} samples. Each of
    these samples is a temporal walk that visits $k$ vertices of $S$ in $O(k(k
    n/|S|+r))$ time steps, computed by a linear-time {consideration} of the snapshots that
    exist in the respective time span {(Algorithm~\ref{alg:one})}. In total, we have a runtime of
    $O((\sum_{t' \in [t, t'']} |E(G_{t''})|) \log(n/\epsilon) n^{1/3})$ for the
    computation of all required samples. 
\end{proof}

{ Using Corollary~\ref{cor:expfraction} we can easily modify the
    previous theorem so that the probability of success is $1$, with the caveat
    of an expected runtime. Instead of using Lemma~\ref{lem:randlog} to
    construct our random walks in each of the $O(n^{2/3}\log n)$ phases of the
    previous theorem, we use Corollary~\ref{cor:expfraction} instead. 

\begin{corollary}\label{cor:expectedorbitexplore} 
    Let $\mathcal{G}$ be a temporal graph with lifetime $\ell$ and vertex set
    $V$. Take $S \in \torb$ and $r=|\torb|$ the orbit number. For any
    $t\in[\ell]$, there exists a temporal walk $W$ that starts at time step
    $t$, visits all vertices of $S$ and ends at time step
    $t'=t+O((n^{5/3}+rn)\log n)$. The temporal walk $W$ can be computed 
    in
    $O((\sum_{t'' \in [t, t']} |E(G_{t''})|) n^{1/3})$ expected time, {i.e.,
    in $O( M_{t,t'} \ n^{1/3})$ expected time where $M_{t,t'}$ is the total number
    of time edges of snapshots $G_t, \ldots, G_{t'}$.}
\end{corollary}
\begin{proof}
    In each phase we expect to perform $O(1/\log(1/(1-p/2)))$ linear-time
    {considerations} of the snapshots that exist during the respective phases. With the
    bound $p \geq 2/n^{1/3}$ described in the proof of the previous theorem,
    this can be simplified to $O(n^{1/3})$, which will be shown in more detail
    below. The arguments follow the equation used in the proof of the previous
    theorem, i.e., we use the Mercator series.
    \begin{align*}
        1/\log(1/(1-p/2))                     &= (-1)/\log(1-p/2)\\
                                            &< (-1)/(-p/2)\\
                                            &\leq (-1)/(-1/n^{1/3})\\
                                            &=  n^{1/3}
    \end{align*}
\end{proof}
}
}
\section{Conclusions \& Future Work}
In this work, we looked at temporal graphs where agents know the complete
information of the temporal graph ahead of time. In this clairvoyant setting,
we studied the temporal exploration problem ($\texp$) and showed how to bound
the exploration time of a temporal graph using the structural graph property of
the number of orbits of the automorphism group of the temporal graph.
Additionally, we formalized the problem of asymmetric rendezvous in this
setting as the temporal rendezvous problem ($\trp$) and showed how to adapt our
ideas for $\texp$ to solve $\trp$ quickly. For both $\texp$ and $\trp$ we
provided lower bounds such that the gap between upper and lower bounds is
$O(n^\epsilon)$ for any fixed $\epsilon > 0$. 
\full{Finally, we gave a randomized algorithm to construct a temporal
    walk $W$ that visits all vertices of a given orbit with probability at least $1-\epsilon$
    for any $0<\epsilon<1$
    such that $W$ spans $O((n^{5/3}+rn)\log n)$ time steps. The runtime of this
    algorithm consists of $O(n^{1/3} \log (n/\epsilon))$ linear-time {considerations} of the
snapshots that exist in the aforementioned time span. For this, we also gave an alternative
algorithm that has an expected runtime that consists of $O(n^{1/3})$ linear-time {considerations},
and provides the specified temporal walk with probability $1$.}

There are several ways in which our work can be extended. One line of research
for both problems is to reduce the gap between the lower and upper bounds by
improving either of them. A second line of work is to study the symmetric variant
of rendezvous in the given setting and see if something can be said about it. Another
interesting situation to explore is when multiple agents are used to explore the
temporal graph (and also if multiple agents need to perform temporal rendezvous)
and how much faster solutions in these scenarios might be.
Lastly, a possible avenue of research is to study the structural properties
provided by automorphism groups and how they can be used to tackle other problems
that concern temporal graphs.
\bibliography{main}

\begin{thebibliography}{10}

\bibitem{AaronKM14a}
Eric Aaron, Danny Krizanc, and Elliot Meyerson.
\newblock Dmvp: foremost waypoint coverage of time-varying graphs.
\newblock In {\em International Workshop on Graph-Theoretic Concepts in
  Computer Science}, pages 29--41. Springer, 2014.
\newblock \href {https://doi.org/10.1007/978-3-319-12340-0_3}
  {\path{doi:10.1007/978-3-319-12340-0_3}}.

\bibitem{AaronKM14b}
Eric Aaron, Danny Krizanc, and Elliot Meyerson.
\newblock Multi-robot foremost coverage of time-varying graphs.
\newblock In {\em International Symposium on Algorithms and Experiments for
  Sensor Systems, Wireless Networks and Distributed Robotics}, pages 22--38.
  Springer, 2014.
\newblock \href {https://doi.org/10.1007/978-3-662-46018-4_2}
  {\path{doi:10.1007/978-3-662-46018-4_2}}.

\bibitem{AdamsonGMZ22}
Duncan Adamson, Vladimir~V. Gusev, Dmitriy Malyshev, and Viktor Zamaraev.
\newblock Faster exploration of some temporal graphs.
\newblock In {\em 1st Symposium on Algorithmic Foundations of Dynamic Networks
  (SAND 2022)}. Schloss Dagstuhl-Leibniz-Zentrum f{\"u}r Informatik, 2022.
\newblock \href {https://doi.org/10.4230/LIPIcs.SAND.2022.5}
  {\path{doi:10.4230/LIPIcs.SAND.2022.5}}.

\bibitem{AKRIDA201946}
Eleni~C. Akrida, Jurek Czyzowicz, Leszek G{\k a}sieniec, {\L}ukasz Kuszner, and
  Paul~G. Spirakis.
\newblock Temporal flows in temporal networks.
\newblock {\em Journal of Computer and System Sciences}, 103:46--60, 2019.
\newblock \href {https://doi.org/10.1016/j.jcss.2019.02.003}
  {\path{doi:10.1016/j.jcss.2019.02.003}}.

\bibitem{AkridaMSR21}
Eleni~C. Akrida, George~B. Mertzios, Paul~G. Spirakis, and Christoforos
  Raptopoulos.
\newblock The temporal explorer who returns to the base.
\newblock {\em Journal of Computer and System Sciences}, 120:179--193, 2021.
\newblock \href {https://doi.org/10.1016/j.jcss.2021.04.001}
  {\path{doi:10.1016/j.jcss.2021.04.001}}.

\bibitem{AKRIDA2020108}
Eleni~C. Akrida, George~B. Mertzios, Paul~G. Spirakis, and Viktor Zamaraev.
\newblock Temporal vertex cover with a sliding time window.
\newblock {\em Journal of Computer and System Sciences}, 107:108--123, 2020.
\newblock \href {https://doi.org/10.1016/j.jcss.2019.08.002}
  {\path{doi:10.1016/j.jcss.2019.08.002}}.

\bibitem{Alpern95}
Steve Alpern.
\newblock The rendezvous search problem.
\newblock {\em SIAM Journal on Control and Optimization}, 33(3):673--683, 1995.
\newblock \href {https://doi.org/10.1137/S0363012993249195}
  {\path{doi:10.1137/S0363012993249195}}.

\bibitem{Alpern02}
Steve Alpern.
\newblock Rendezvous search: A personal perspective.
\newblock {\em Operations Research}, 50(5):772--795, 2002.
\newblock \href {https://doi.org/10.1287/opre.50.5.772.363}
  {\path{doi:10.1287/opre.50.5.772.363}}.

\bibitem{AlpernG06}
Steve Alpern and Shmuel Gal.
\newblock {\em The theory of search games and rendezvous}, volume~55 of {\em
  International Series in Operations Research \& Management Science}.
\newblock Springer Science \& Business Media, 2006.
\newblock \href {https://doi.org/10.1007/b100809} {\path{doi:10.1007/b100809}}.

\bibitem{Alpern76}
Steven Alpern.
\newblock Hide and seek games.
\newblock In {\em Seminar, Institut f{\"u}r h{\"o}here Studien, Wien},
  volume~26, 1976.

\bibitem{Balasubramanian82}
K.~Balasubramanian.
\newblock Symmetry groups of chemical graphs.
\newblock {\em International Journal of Quantum Chemistry}, 21(2):411--418,
  1982.
\newblock \href {https://doi.org/10.1002/qua.560210206}
  {\path{doi:10.1002/qua.560210206}}.

\bibitem{BallGS18}
Fabian Ball and Andreas Geyer-Schulz.
\newblock How symmetric are real-world graphs? a large-scale study.
\newblock {\em Symmetry}, 10(1), 2018.
\newblock \href {https://doi.org/10.3390/sym10010029}
  {\path{doi:10.3390/sym10010029}}.

\bibitem{BhagatP22}
Subhash Bhagat and Andrzej Pelc.
\newblock Deterministic rendezvous in infinite trees.
\newblock {\em CoRR}, abs/2203.05160, 2022.
\newblock \href {http://arxiv.org/abs/2203.05160} {\path{arXiv:2203.05160}},
  \href {https://doi.org/10.48550/arXiv.2203.05160}
  {\path{doi:10.48550/arXiv.2203.05160}}.

\bibitem{BodlaenderV19}
Hans~L. Bodlaender and Tom~C. van~der Zanden.
\newblock On exploring always-connected temporal graphs of small pathwidth.
\newblock {\em Information Processing Letters}, 142:68--71, 2019.
\newblock \href {https://doi.org/10.1016/j.ipl.2018.10.016}
  {\path{doi:10.1016/j.ipl.2018.10.016}}.

\bibitem{BondyM76}
John~Adrian Bondy and Uppaluri Siva~Ramachandra Murty.
\newblock {\em Graph theory with applications}, volume 290.
\newblock Macmillan London, 1976.

\bibitem{BournatDP18}
Marjorie Bournat, Swan Dubois, and Franck Petit.
\newblock Gracefully degrading gathering in dynamic rings.
\newblock In {\em Stabilization, Safety, and Security of Distributed Systems:
  20th International Symposium, SSS 2018, Tokyo, Japan, November 4--7, 2018,
  Proceedings 20}, pages 349--364. Springer, 2018.
\newblock \href {https://doi.org/10.1007/978-3-030-03232-6_23}
  {\path{doi:10.1007/978-3-030-03232-6_23}}.

\bibitem{BumpusM23}
Benjamin~Merlin Bumpus and Kitty Meeks.
\newblock Edge exploration of temporal graphs.
\newblock {\em Algorithmica}, 85(3):688--716, 2023.
\newblock \href {https://doi.org/10.1007/s00453-022-01018-7}
  {\path{doi:10.1007/s00453-022-01018-7}}.

\bibitem{Casteigts20}
Arnaud Casteigts.
\newblock Efficient generation of simple temporal graphs up to isomorphism.
\newblock Github repository, 2020.
\newblock URL: \url{https://github.com/acasteigts/STGen}.

\bibitem{CasteigtsFQS12}
Arnaud Casteigts, Paola Flocchini, Walter Quattrociocchi, and Nicola Santoro.
\newblock Time-varying graphs and dynamic networks.
\newblock {\em Int. J. Parallel Emergent Distributed Syst.}, 27(5):387--408,
  2012.
\newblock \href {https://doi.org/10.1080/17445760.2012.668546}
  {\path{doi:10.1080/17445760.2012.668546}}.

\bibitem{CzyzowiczJKAP10}
Jurek Czyzowicz, Adrian Kosowski, and Andrzej Pelc.
\newblock How to meet when you forget: Log-space rendezvous in arbitrary
  graphs.
\newblock In {\em Proceedings of the 29th ACM SIGACT-SIGOPS Symposium on
  Principles of Distributed Computing}, PODC '10, page 450–459, New York, NY,
  USA, 2010. Association for Computing Machinery.
\newblock \href {https://doi.org/10.1145/1835698.1835801}
  {\path{doi:10.1145/1835698.1835801}}.

\bibitem{DasDPP19}
Shantanu Das, Giuseppe Di~Luna, Linda Pagli, and Giuseppe Prencipe.
\newblock Compacting and grouping mobile agents on dynamic rings.
\newblock In {\em International Conference on Theory and Applications of Models
  of Computation}, pages 114--133. Springer, 2019.
\newblock \href {https://doi.org/10.1007/978-3-030-14812-6_8}
  {\path{doi:10.1007/978-3-030-14812-6_8}}.

\bibitem{DiLuna19}
Giuseppe~Antonio Di~Luna.
\newblock Mobile agents on dynamic graphs.
\newblock {\em Distributed Computing by Mobile Entities: Current Research in
  Moving and Computing}, pages 549--584, 2019.
\newblock \href {https://doi.org/10.1007/978-3-030-11072-7_20}
  {\path{doi:10.1007/978-3-030-11072-7_20}}.

\bibitem{DiLunaFPPSV20}
Giuseppe~Antonio Di~Luna, Paola Flocchini, Linda Pagli, Giuseppe Prencipe,
  Nicola Santoro, and Giovanni Viglietta.
\newblock Gathering in dynamic rings.
\newblock {\em Theoretical Computer Science}, 811:79--98, 2020.
\newblock \href {https://doi.org/10.1016/j.tcs.2018.10.018}
  {\path{doi:10.1016/j.tcs.2018.10.018}}.

\bibitem{dieudonne2013meet}
Yoann Dieudonn{\'e}, Andrzej Pelc, and Vincent Villain.
\newblock How to meet asynchronously at polynomial cost.
\newblock In {\em Proceedings of the 2013 ACM Symposium on Principles of
  Distributed Computing}, pages 92--99, 2013.
\newblock \href {https://doi.org/10.1137/130931990}
  {\path{doi:10.1137/130931990}}.

\bibitem{ENRIGHT202160}
Jessica Enright, Kitty Meeks, George~B. Mertzios, and Viktor Zamaraev.
\newblock Deleting edges to restrict the size of an epidemic in temporal
  networks.
\newblock {\em Journal of Computer and System Sciences}, 119:60--77, 2021.
\newblock \href {https://doi.org/10.1016/j.jcss.2021.01.007}
  {\path{doi:10.1016/j.jcss.2021.01.007}}.

\bibitem{ErlebachHK21}
Thomas Erlebach, Michael Hoffmann, and Frank Kammer.
\newblock On temporal graph exploration.
\newblock {\em J. Comput. Syst. Sci.}, 119:1--18, 2021.
\newblock \href {https://doi.org/10.1016/j.jcss.2021.01.005}
  {\path{doi:10.1016/j.jcss.2021.01.005}}.

\bibitem{ErlebachKLSS19}
Thomas Erlebach, Frank Kammer, Kelin Luo, Andrej Sajenko, and Jakob~T. Spooner.
\newblock Two moves per time step make a difference.
\newblock In {\em 46th International Colloquium on Automata, Languages, and
  Programming (ICALP 2019)}. Schloss Dagstuhl-Leibniz-Zentrum f{\"u}r
  Informatik, 2019.
\newblock \href {https://doi.org/10.4230/LIPIcs.ICALP.2019.141}
  {\path{doi:10.4230/LIPIcs.ICALP.2019.141}}.

\bibitem{Erlebach-Spooner/20}
Thomas Erlebach and Jakob~T. Spooner.
\newblock A game of cops and robbers on graphs with periodic edge-connectivity.
\newblock In {\em 46th International Conference on Current Trends in Theory and
  Practice of Informatics (SOFSEM 2020)}, volume 12011 of {\em Lecture Notes in
  Computer Science}, pages 64--75. Springer, 2020.
\newblock \href {https://doi.org/10.1007/978-3-030-38919-2\_6}
  {\path{doi:10.1007/978-3-030-38919-2\_6}}.

\bibitem{ErlebachS22a}
Thomas Erlebach and Jakob~T. Spooner.
\newblock Exploration of k-edge-deficient temporal graphs.
\newblock {\em Acta Informatica}, 59(4):387--407, 2022.
\newblock \href {https://doi.org/10.1007/s00236-022-00421-5}
  {\path{doi:10.1007/s00236-022-00421-5}}.

\bibitem{ErlebachS22b}
Thomas Erlebach and Jakob~T. Spooner.
\newblock Parameterized temporal exploration problems.
\newblock In {\em 1st Symposium on Algorithmic Foundations of Dynamic Networks,
  {SAND} 2022, March 28-30, 2022, Virtual Conference}, volume 221 of {\em
  LIPIcs}, pages 15:1--15:17. Schloss Dagstuhl - Leibniz-Zentrum f{\"{u}}r
  Informatik, 2022.
\newblock \href {https://doi.org/10.4230/LIPIcs.SAND.2022.15}
  {\path{doi:10.4230/LIPIcs.SAND.2022.15}}.

\bibitem{Flocchini19}
Paola Flocchini.
\newblock {\em Distributed Computing by Mobile Entities: Current Research in
  Moving and Computing}.
\newblock Springer, 2019.
\newblock \href {https://doi.org/10.1007/978-3-030-11072-7}
  {\path{doi:10.1007/978-3-030-11072-7}}.

\bibitem{Fluschnik2020}
Till Fluschnik, Hendrik Molter, Rolf Niedermeier, Malte Renken, and Philipp
  Zschoche.
\newblock As time goes by: Reflections on treewidth for temporal graphs.
\newblock In {\em Treewidth, Kernels, and Algorithms: Essays Dedicated to Hans
  L. Bodlaender on the Occasion of His 60th Birthday}, pages 49--77. Springer
  International Publishing, 2020.
\newblock \href {https://doi.org/10.1007/978-3-030-42071-0_6}
  {\path{doi:10.1007/978-3-030-42071-0_6}}.

\bibitem{FLUSCHNIK2020197}
Till Fluschnik, Hendrik Molter, Rolf Niedermeier, Malte Renken, and Philipp
  Zschoche.
\newblock Temporal graph classes: A view through temporal separators.
\newblock {\em Theoretical Computer Science}, 806:197--218, 2020.
\newblock \href {https://doi.org/10.1016/j.tcs.2019.03.031}
  {\path{doi:10.1016/j.tcs.2019.03.031}}.

\bibitem{GodsilR01}
Chris Godsil and Gordon~F. Royle.
\newblock {\em Algebraic Graph Theory}.
\newblock Number Book 207 in Graduate Texts in Mathematics. Springer, 2001.
\newblock \href {https://doi.org/10.1007/978-1-4613-0163-9}
  {\path{doi:10.1007/978-1-4613-0163-9}}.

\bibitem{HammKMS22}
Thekla Hamm, Nina Klobas, George~B. Mertzios, and Paul~G. Spirakis.
\newblock The complexity of temporal vertex cover in small-degree graphs.
\newblock {\em Proceedings of the AAAI Conference on Artificial Intelligence},
  36(9):10193--10201, Jun. 2022.
\newblock \href {https://doi.org/10.1609/aaai.v36i9.21259}
  {\path{doi:10.1609/aaai.v36i9.21259}}.

\bibitem{IlcinkasKW14}
David Ilcinkas, Ralf Klasing, and Ahmed~Mouhamadou Wade.
\newblock Exploration of constantly connected dynamic graphs based on cactuses.
\newblock In {\em International Colloquium on Structural Information and
  Communication Complexity}, pages 250--262. Springer, 2014.
\newblock \href {https://doi.org/10.1007/978-3-319-09620-9_20}
  {\path{doi:10.1007/978-3-319-09620-9_20}}.

\bibitem{IlcinkasW13}
David Ilcinkas and Ahmed~Mouhamadou Wade.
\newblock Exploration of the t-interval-connected dynamic graphs: the case of
  the ring.
\newblock In {\em International Colloquium on Structural Information and
  Communication Complexity}, pages 13--23. Springer, 2013.
\newblock \href {https://doi.org/10.1007/978-3-319-03578-9_2}
  {\path{doi:10.1007/978-3-319-03578-9_2}}.

\bibitem{KnauerK19}
Ulrich Knauer and Kolja Knauer.
\newblock {\em Algebraic graph theory: morphisms, monoids and matrices},
  volume~41.
\newblock Walter de Gruyter GmbH \& Co KG, 2019.
\newblock \href {https://doi.org/10.1515/9783110617368}
  {\path{doi:10.1515/9783110617368}}.

\bibitem{KranakisKM22}
Evangelos Kranakis, Danny Krizanc, and Euripides Marcou.
\newblock {\em The mobile agent rendezvous problem in the ring}.
\newblock Springer Nature, 2022.
\newblock \href {https://doi.org/10.1007/978-3-031-01999-9}
  {\path{doi:10.1007/978-3-031-01999-9}}.

\bibitem{KranakisKR06}
Evangelos Kranakis, Danny Krizanc, and Sergio Rajsbaum.
\newblock Mobile agent rendezvous: A survey.
\newblock In {\em Structural Information and Communication Complexity}, pages
  1--9. Springer Berlin Heidelberg, 2006.
\newblock \href {https://doi.org/10.1007/11780823_1}
  {\path{doi:10.1007/11780823_1}}.

\bibitem{LauriS16}
Josef Lauri and Raffaele Scapellato.
\newblock {\em Topics in Graph Automorphisms and Reconstruction}.
\newblock Cambridge University Press, 2016.
\newblock \href {https://doi.org/10.1017/CBO9781316669846}
  {\path{doi:10.1017/CBO9781316669846}}.

\bibitem{MARCO2006315}
Gianluca~De Marco, Luisa Gargano, Evangelos Kranakis, Danny Krizanc, Andrzej
  Pelc, and Ugo Vaccaro.
\newblock Asynchronous deterministic rendezvous in graphs.
\newblock {\em Theoretical Computer Science}, 355(3):315--326, 2006.
\newblock \href {https://doi.org/10.1016/j.tcs.2005.12.016}
  {\path{doi:10.1016/j.tcs.2005.12.016}}.

\bibitem{MarinoS21}
Andrea Marino and Ana Silva.
\newblock K{\"o}nigsberg sightseeing: Eulerian walks in temporal graphs.
\newblock In {\em Combinatorial Algorithms}, pages 485--500. Springer
  International Publishing, 2021.
\newblock \href {https://doi.org/10.1007/978-3-030-79987-8_34}
  {\path{doi:10.1007/978-3-030-79987-8_34}}.

\bibitem{MarinoS22}
Andrea Marino and Ana Silva.
\newblock Coloring temporal graphs.
\newblock {\em Journal of Computer and System Sciences}, 123:171--185, 2022.
\newblock \href {https://doi.org/10.1016/j.jcss.2021.08.004}
  {\path{doi:10.1016/j.jcss.2021.08.004}}.

\bibitem{MertziosMNZZ20}
George~B. Mertzios, Hendrik Molter, Rolf Niedermeier, Viktor Zamaraev, and
  Philipp Zschoche.
\newblock {Computing Maximum Matchings in Temporal Graphs}.
\newblock In {\em 37th International Symposium on Theoretical Aspects of
  Computer Science (STACS 2020)}, volume 154 of {\em Leibniz International
  Proceedings in Informatics (LIPIcs)}, pages 27:1--27:14, Dagstuhl, Germany,
  2020. Schloss Dagstuhl--Leibniz-Zentrum f{\"u}r Informatik.
\newblock \href {https://doi.org/10.4230/LIPIcs.STACS.2020.27}
  {\path{doi:10.4230/LIPIcs.STACS.2020.27}}.

\bibitem{MertziosMZ21}
George~B. Mertzios, Hendrik Molter, and Viktor Zamaraev.
\newblock Sliding window temporal graph coloring.
\newblock {\em Journal of Computer and System Sciences}, 120:97--115, 2021.
\newblock \href {https://doi.org/10.1016/j.jcss.2021.03.005}
  {\path{doi:10.1016/j.jcss.2021.03.005}}.

\bibitem{Michail16}
Othon Michail.
\newblock An introduction to temporal graphs: An algorithmic perspective.
\newblock {\em Internet Math.}, 12(4):239--280, 2016.
\newblock \href {https://doi.org/10.1080/15427951.2016.1177801}
  {\path{doi:10.1080/15427951.2016.1177801}}.

\bibitem{MichailS14}
Othon Michail and Paul~G. Spirakis.
\newblock Traveling salesman problems in temporal graphs.
\newblock {\em Computer Science (MFCS)}, 2014.
\newblock \href {https://doi.org/10.1016/j.tcs.2016.04.006}
  {\path{doi:10.1016/j.tcs.2016.04.006}}.

\bibitem{MichailS16}
Othon Michail and Paul~G. Spirakis.
\newblock Traveling salesman problems in temporal graphs.
\newblock {\em Theoretical Computer Science}, 634:1--23, 2016.
\newblock \href {https://doi.org/10.1016/j.tcs.2016.04.006}
  {\path{doi:10.1016/j.tcs.2016.04.006}}.

\bibitem{MichailST21}
Othon Michail, Paul~G. Spirakis, and Michail Theofilatos.
\newblock Beyond rings: Gathering in 1-interval connected graphs.
\newblock {\em Parallel Processing Letters}, 31(04):2150020, 2021.
\newblock \href {https://doi.org/10.1142/S0129626421500201}
  {\path{doi:10.1142/S0129626421500201}}.

\bibitem{MorawietzRW20}
Nils Morawietz, Carolin Rehs, and Mathias Weller.
\newblock {A Timecop’s Work Is Harder Than You Think}.
\newblock In {\em 45th International Symposium on Mathematical Foundations of
  Computer Science (MFCS 2020)}, volume 170 of {\em Leibniz International
  Proceedings in Informatics (LIPIcs)}, pages 71:1--71:14, Dagstuhl, Germany,
  2020. Schloss Dagstuhl--Leibniz-Zentrum f{\"u}r Informatik.
\newblock \href {https://doi.org/10.4230/LIPIcs.MFCS.2020.71}
  {\path{doi:10.4230/LIPIcs.MFCS.2020.71}}.

\bibitem{OoshitaD18}
Fukuhito Ooshita and Ajoy~K. Datta.
\newblock Brief announcement: feasibility of weak gathering in
  connected-over-time dynamic rings.
\newblock In {\em Stabilization, Safety, and Security of Distributed Systems:
  20th International Symposium, SSS 2018, Tokyo, Japan, November 4--7, 2018,
  Proceedings 20}, pages 393--397. Springer, 2018.
\newblock \href {https://doi.org/10.1007/978-3-030-03232-6_27}
  {\path{doi:10.1007/978-3-030-03232-6_27}}.

\bibitem{Pelc12}
Andrzej Pelc.
\newblock Deterministic rendezvous in networks: A comprehensive survey.
\newblock {\em Networks}, 59(3):331--347, 2012.
\newblock \href {https://doi.org/10.1002/net.21453}
  {\path{doi:10.1002/net.21453}}.

\bibitem{Pelc19}
Andrzej Pelc.
\newblock Deterministic rendezvous algorithms.
\newblock In {\em Distributed Computing by Mobile Entities: Current Research in
  Moving and Computing}, pages 423--454. Springer, 2019.
\newblock \href {https://doi.org/10.1007/978-3-030-11072-7\_17}
  {\path{doi:10.1007/978-3-030-11072-7\_17}}.

\bibitem{Pelc23}
Andrzej Pelc.
\newblock Deterministic rendezvous algorithms.
\newblock {\em CoRR}, abs/2303.10391, 2023.
\newblock \href {http://arxiv.org/abs/2303.10391} {\path{arXiv:2303.10391}},
  \href {https://doi.org/10.48550/arXiv.2303.10391}
  {\path{doi:10.48550/arXiv.2303.10391}}.

\bibitem{PlamperLH23}
Philipp Plamper, Oliver~J. Lechtenfeld, Peter Herzsprung, and Anika Groß.
\newblock A temporal graph model to predict chemical transformations in complex
  dissolved organic matter.
\newblock {\em Environmental Science \& Technology}, 2023.
\newblock \href {https://doi.org/10.1021/acs.est.3c00351}
  {\path{doi:10.1021/acs.est.3c00351}}.

\bibitem{Santoro15}
Nicola Santoro.
\newblock Time to change: On distributed computing in dynamic networks
  (keynote).
\newblock In {\em 19th International Conference on Principles of Distributed
  Systems, {OPODIS} 2015, December 14-17, 2015, Rennes, France}, volume~46 of
  {\em LIPIcs}, pages 3:1--3:14. Schloss Dagstuhl - Leibniz-Zentrum f{\"{u}}r
  Informatik, 2015.
\newblock \href {https://doi.org/10.4230/LIPIcs.OPODIS.2015.3}
  {\path{doi:10.4230/LIPIcs.OPODIS.2015.3}}.

\bibitem{Shannon51}
C.~Shannon.
\newblock Presentation of a maze solving machine.
\newblock In {\em Trans. 8th Conf. Cybernetics: Circular, Causal and Feedback
  Mechanisms in Biological and Social Systems (New York, 1951)}, pages
  169--181, 1951.

\bibitem{ShibataKESNK23}
Masahiro Shibata, Naoki Kitamura, Ryota Eguchi, Yuichi Sudo, Junya Nakamura,
  and Yonghwan Kim.
\newblock Partial gathering of mobile agents in dynamic tori.
\newblock In {\em 2nd Symposium on Algorithmic Foundations of Dynamic Networks
  (SAND 2023)}. Schloss Dagstuhl-Leibniz-Zentrum f{\"u}r Informatik, 2023.
\newblock \href {https://doi.org/10.4230/LIPIcs.SAND.2023.2}
  {\path{doi:10.4230/LIPIcs.SAND.2023.2}}.

\bibitem{ShibataSNK21}
Masahiro Shibata, Yuichi Sudo, Junya Nakamura, and Yonghwan Kim.
\newblock Partial gathering of mobile agents in dynamic rings.
\newblock In {\em Stabilization, Safety, and Security of Distributed Systems:
  23rd International Symposium, SSS 2021, Virtual Event, November 17--20, 2021,
  Proceedings 23}, pages 440--455. Springer, 2021.
\newblock \href {https://doi.org/10.1007/978-3-030-91081-5_29}
  {\path{doi:10.1007/978-3-030-91081-5_29}}.

\bibitem{Taghian20}
Shadi Taghian~Alamouti.
\newblock {\em Exploring temporal cycles and grids}.
\newblock PhD thesis, Concordia University, 2020.

\bibitem{WuCYSYW16}
Huanhuan Wu, James Cheng, Yiping Ke, Silu Huang, Yuzhen Huang, and Hejun Wu.
\newblock Efficient algorithms for temporal path computation.
\newblock {\em IEEE Transactions on Knowledge and Data Engineering},
  28(11):2927--2942, 2016.
\newblock \href {https://doi.org/10.1109/TKDE.2016.2594065}
  {\path{doi:10.1109/TKDE.2016.2594065}}.

\bibitem{ZSCHOCHE202072}
Philipp Zschoche, Till Fluschnik, Hendrik Molter, and Rolf Niedermeier.
\newblock The complexity of finding small separators in temporal graphs.
\newblock {\em Journal of Computer and System Sciences}, 107:72--92, 2020.
\newblock \href {https://doi.org/10.1016/j.jcss.2019.07.006}
  {\path{doi:10.1016/j.jcss.2019.07.006}}.

\end{thebibliography}
\app{
\newpage
\appendix
\setcounter{theorem}{1}
\section{Same Number of Automorphisms}
\begin{lemma}
    Let $S \in \torb$ be any orbit. Then
    $|\auttg[u,{x_1}]|=|\auttg[u,{x_2}]|$ for any $u, x_1, x_2 \in
    S$.
\end{lemma}
\proofnumautos

\section{Applying Automorphisms to Temporal Walks}
\begin{lemma}\label{lem:transform1} Let $\mathcal{G}$ be a connected temporal
    graph with lifetime $\ell$ and let $S \in \torb$ be any orbit. Let $T
    \subseteq S$ with $p=(|S|-|T|)/|S|$ and $W=(u_1, u_2, \ldots, u_x)$ a
    temporal walk starting at time step $t$ and ending at $t'$ with $t, t' \in
    [\ell]$ and $u_1 \in S$ such that $W$ visits $k$ vertices of $S$. Then there
    exists a temporal walk $W'$ starting at a vertex $u' \in S$ in time step $t$
    and ending at time step $t'$ that visits at least $p k$ vertices of $S
    \setminus T$.
\end{lemma}
\prooftransformone

\section{Applying Automorphisms with Restricted Start Vertices}
\begin{corollary} Let $\mathcal{G}$ be a connected
    temporal graph with lifetime $\ell$ and let $S \in \torb$ be any orbit. Let
    $T \subsetneq S$ and $W$ a temporal walk starting at time step $t$ and
    ending at $t'$ with $t, t' \in [\ell]$ such that the first vertex of $W$ is
    in $S$ and such that $W$ visits $k$ different vertices of $S$. For any $X
    \subseteq S$ with $|X| > |T|$ there exists a temporal walk $W'$ starting at
    a vertex $u' \in X$ at time step $t$ and ending at time step $t'$ that
    visits at least $(c-1)/c \cdot k$ vertices of $S \setminus T$, with
    $c=|X|/|T|$.
\end{corollary}
\prooftransformtwo

\section{Vertices in the Same Orbit and Their Incident Boundary Edges}
\begin{lemma}\label{lem:orbitedges} {Let $G_t$ be the graph at time
    step~$t$ in a connected temporal graph~$\mathcal{G}$ and $S,S' \in \torb$,
    and let $G'$ be the subgraph of $G_t$ that contains only orbit boundary
    edges. Then all vertices in $S$ have the same degree in the bipartite graph
    $G'[S \cup S']$.}
\end{lemma}
\prooforbitedges

\section{Reachability between Orbits}
\begin{lemma} Let
    $\mathcal{G}$ be a connected temporal graph with lifetime $\ell$ and $S \in
    \torb$. For any $X \subseteq S$ and $S' \in \torb$ it holds that $|L_{t,
    t'}(X) \cap S'| \geq  \lceil |X| \cdot |S'|/|S| \rceil$ for any $t \in
    [\ell]$ and $t'=t+r$, where $r=|\torb|$ is the orbit number.
\end{lemma}
\prooflaneorbit

\section{Reaching $h$ Vertices inside a Single Orbit}
\begin{lemma} 
Let $\mathcal{G}$ be a connected temporal graph with lifetime $\ell$ and vertex
set $V$. Let $S \in \torb$ and let $r=|\torb|$ be the orbit number. For any $h
\leq |S|$, start vertex {$u \in S$} {and start time $t$}, there exists a set $X
\subset S$ {with $|X|=h$} such that we can reach any vertex in $X$ in at most
$O(\min\{h \cdot n / |S|, h r\}+r)$ time steps. That is, for every vertex $u'$
of $X$, we have a temporal walk starting at $u$ at time step $t$ and ending at
$u'$ at time step $t'$ with $t'-t=O(\min\{h \cdot n / |S|, h r\}+r)$.
 \end{lemma}
\proofnextreachable

\section{Visiting All Vertices inside a Single Orbit}
\begin{theorem} Let $\mathcal{G}$ be a temporal graph
    with lifetime $\ell$ and vertex set $V$. Take $S \in \torb$ and $r=|\torb|$
    the orbit number. For any $t\in[\ell]$ there exists a temporal walk $W$
    starting at time step $t$ that visits all vertices of $S$ and ends at time
    step $t'$ with $t'-t=O((n^{5/3}+rn)\log n)$. 
\end{theorem}
\prooforbitexplore

\section{Visiting a Guaranteed Fraction of Vertices in an Orbit}
\begin{lemma} Let $\mathcal{G}$ be a connected
    temporal graph with lifetime $\ell$ and vertex set $V$. Let $S \in \torb$
    and let $r=|\torb|$ be the orbit number. For any $t\in[\ell]$ and any
    {$\ww \in S$} there exists a temporal walk $W$ that starts at vertex
    $\ww$ in time step $t$ and visits a fraction $1/c$ (for any $1 < c <\
    |S|$) of the vertices of $S$ such that $W$ spans $O(rc(|S|/c)^{\phi(c)}
    \log |S|)$ time steps, with $\phi(c)=1/(\log {f(c)})$ and
    $f(c)=(1+(c-1)/c)$.
\end{lemma}
\proofvtexplorationub

\section{Visiting a Constant Fraction of Vertices in an Orbit}
\begin{corollary}\label{cor:tinyfraction} Let $\mathcal{G}$ be a temporal graph
    with lifetime $\ell$ and vertex set $V$. Let $S \in \torb$ and $r=|\torb|$
    be the orbit number. For any $t\in[\ell]$, any $u \in S$, and any fixed
    $\epsilon > 0$,  there exists a temporal walk $W$ that starts at vertex $u$
    in time step $t$ and visits some constant fraction $\alpha < 1$  
    of the vertices of
    $S$ such that 
    $W$ spans $O(r|S|^{1+\epsilon})$ time steps.     
\end{corollary}
\prooftinyfraction

\section{Visiting all Vertices inside a Single Orbit Faster than Theorem~\ref{thm:orbitexplore}}
\begin{theorem} Let $\mathcal{G}$ be a temporal
    graph with lifetime $\ell$ and vertex set $V$. Let $S \in \torb$ and
    $r=|\torb|$ be the orbit number. For any $t\in[\ell]$, any $u \in V$, and
    any fixed $\epsilon > 0$, there exists a temporal walk $W$ that starts at
    vertex $u$ in time step $t$ and visits all vertices of $S$ such that $W$
    spans $O(r|S|^{1+\epsilon} + n \log |S|)$ time steps.
\end{theorem}
\proofepsilonexploration

\section{Visiting all Vertices of a Temporal Graph using Theorem~\ref{thm:epsilonexploration}}
\begin{corollary} Let $\mathcal{G}$ be a temporal
    graph with lifetime $\ell$ and vertex set $V$.
    For any fixed $\epsilon>0$, there exists a temporal walk
    $W$ that spans $O(r n^{1+\epsilon})$ time steps and visits all vertices of
    $V$, where $r=|\torb|$ is the orbit number.
\end{corollary}
\proofepsilonallorbits

\section{A Canonical Orbit to Meet}
\begin{lemma}\label{lem:meeting} Let $\mathcal{G}$ be a temporal graph with
    vertex set $V$ and lifetime $\ell$ and let $a_1, a_2$ be two label-oblivious
    agents. There exists a pair of programs $(p_1, p_2)$ assigned to $a_1$ and
    $a_2$, respectively, such that each agent computes the same integer coloring
    of $V$ and such that two vertices $u, v \in V$ have the same color exactly
    if $u, v$ are in the same orbit of $\mathcal{G}/{\text{Aut}(\mathcal{G})}$.
\end{lemma}
\proofmeeting

\setcounter{theorem}{14}
\section{Lower Bound for TEXP Based on the Number of Orbits}
\begin{lemma}
        For any $1\le r\le n$, there exist $n$-vertex instances of $\texp$ with
        orbit number $r$ that require $\Omega(rn)$ time steps to be explored. 
    \end{lemma}
\proornlowerbound

\section{Lower Bound for TRP}
\begin{theorem}
    For any two agents $a_1$ and $a_2$ with arbitrary deterministic programs,
    there exist instances of $\trp$ where the agents
    require $\Omega(n \log n)$ time steps to meet.
\end{theorem}
\proofrendezvouslb

\section{Lower Bound for TEXP in Single-Orbit Temporal Graphs}
\begin{corollary} There exist connected temporal graphs
    $\mathcal{G}$ with vertex set $V$, lifetime $\ell$ and a single orbit such
    that all temporal walks $W$ require $\Omega(n \log n)$ time steps to visit
    all vertices of $V$.
\end{corollary}
\proofexplorationlb
}
\end{document}